\documentclass[a4paper,11pt]{article}

\usepackage{fullpage}
\usepackage{times}
\usepackage{soul}
\usepackage{url}
\usepackage[utf8]{inputenc}
\usepackage[small]{caption}
\usepackage{graphicx}
\usepackage{amsmath}
\usepackage{booktabs}

\usepackage[ruled,vlined]{algorithm2e}
\SetArgSty{textrm}

\usepackage{enumitem}

\usepackage{libertine}

\usepackage{amssymb,amsmath,amsfonts,amstext,amsthm}

\usepackage[breaklinks=true]{hyperref}
\usepackage[svgnames]{xcolor}
\usepackage[capitalise,nameinlink]{cleveref}
\hypersetup{colorlinks={true},linkcolor={DarkBlue},citecolor=[named]{DarkGreen}}

\usepackage{tikz}  
\usetikzlibrary{arrows}
\usetikzlibrary{patterns,snakes}
\usetikzlibrary{decorations.shapes}
\tikzstyle{overbrace text style}=[font=\tiny, above, pos=.5, yshift=5pt]
\tikzstyle{overbrace style}=[decorate,decoration={brace,raise=5pt,amplitude=3pt}]
\usetikzlibrary{shapes.geometric}
\usetikzlibrary{shapes,positioning}
\usetikzlibrary{calc,positioning}
\usetikzlibrary{shapes.multipart}

\definecolor{cadmiumgreen}{rgb}{0.0, 0.42, 0.24}

\usepackage{bbm}

\newtheorem{theorem}{Theorem}[section]
\newtheorem{corollary}[theorem]{Corollary}

\theoremstyle{definition}

\newtheorem{example}[theorem]{Example}

\usepackage{natbib}

\usepackage{mathtools}
\usepackage{xcolor} 

\usepackage{authblk}

\usepackage{subcaption}

\usepackage{mdframed}
\mdfsetup{linewidth=1pt}

\usepackage{xspace}
\usepackage{fancyhdr}
\usepackage{tcolorbox}

\newcommand{\alex}[1]{{\color{blue}[Alex: #1]}}
\newcommand{\pk}[1]{{\color{red}[pk: #1]}}

\newcommand{\cost}{\text{cost}}
\newcommand{\opt}{\mathbf{o}}
\newcommand{\bw}{\mathbf{w}}

\newcommand{\bx}{\mathbf{x}}
\newcommand{\bc}{\mathbf{c}}
\newcommand{\bp}{\mathbf{p}}

\newcommand{\SC}{\text{\normalfont SC}}
\newcommand{\MC}{\text{\normalfont MC}}

\newcommand{\alphaMech}{\text{\sc $\alpha$-Statistic}}
\newcommand{\PCM}{\text{\sc Stronger-Majority-Median}}

\newcommand{\maxV}{\text{\sc Vote-for-Priority}}

\allowdisplaybreaks

\title{\bf Truthful Two-Facility Location with Candidate Locations}
\author{Panagiotis Kanellopoulos, Alexandros A. Voudouris, Rongsen Zhang}
\date{School of Computer Science and Electronic Engineering, University of Essex, UK}

\begin{document}

\maketitle

\begin{abstract}
We study a truthful two-facility location problem in which a set of agents have private positions on the line of real numbers and known approval preferences over two different facilities. Given the locations of the two facilities, the cost of an agent is the total distance from the facilities she approves. 
The goal is to decide where to place the facilities from a given finite set of candidate locations so as to 
(a) approximately optimize desired social objectives, and 
(b) incentivize the agents to truthfully report their private positions. 
We focus on the class of deterministic strategyproof mechanisms and show bounds on their approximation ratio in terms of the social cost (i.e., the total cost of the agents) and the max cost for several classes of instances depending on the preferences of the agents over the facilities. 
\end{abstract}

\section{Introduction}
In the well-studied {\em truthful single-facility location problem}, there is a set of agents with private positions on the line of real numbers, and a facility (such as a park or a school) that is to be located somewhere on the line. Given such a location, its distance from the position of an agent is interpreted as the individual cost that the agent would suffer by having to travel to the facility in order to be serviced by it. The goal is to determine the facility location so that the agents are given the right incentives to {\em truthfully} report their positions (that is, not being able to affect the outcome to decrease their cost), and, at the same time, a social function of the individual agent costs (such as the total or the maximum cost) is (approximately) optimized. Since the work of \citet{procaccia09approximate}, who were the first to consider facility location problems through the prism of {\em approximate mechanism design without money}, research on this topic has flourished and a large number of more complex variants of the problem have been introduced and analyzed; see the survey of \citet{fl-survey} for an overview. 

The original work of \citet{procaccia09approximate} focused on a {\em continuous} model, where the facility is allowed to be placed at any point of the line, and showed tight bounds on the approximation ratio of deterministic and randomized {\em strategyproof} mechanisms in terms of the {\em social cost} (total individual cost of the agents) and the {\em max cost} (maximum individual cost among all agents). The {\em discrete} model, where the facility can be placed only at a given set of candidate points of the line, has also been studied, most notably by \citet{feldman2016voting}. They observed that this setting is equivalent to voting on a line, and the strategyproofness constraint leads to deterministic mechanisms that make decisions using only the ordinal preferences of agents over the candidate points; an assumption typically made in the distortion literature~\citep{distortion-survey}. 
Many other truthful single-facility location models have been studied under different assumptions, such as that the location space is more general than a line~\citep{alon2010networks,meir2019circle,goel2020coordinate}, that the facility location must be decided in a distributed way~\citep{FV21,FKVZ23}, or that the facility is obnoxious and the agents aim to be as far from it as possible~\citep{cheng2013obnoxious}. 

Aiming to locate multiple facilities is a natural generalization. Most of the work in this direction has focused on the fundamental case of two facilities, under several different assumptions about the types of the facilities, the preferences of the agents for them, the agent-related information that is public or private, and whether the setting is continuous or discrete, aiming to capture different applications. 
When the facilities are of the same type and serve the same purpose (for example, they might correspond to two parks), with only a few exceptions, the typical assumption is that each agent cares about one of them, such as the facility placed closest to her position (i.e., the agent would like to have a shortest path to a facility) \citep{procaccia09approximate,Lu2010two-facility,fotakis2016concave}, or the facility farthest from her position (i.e., the agent would like to be in a short radius from both facilities)~\citep{chen2020optional}. In most works the facilities are considered desirable, but settings in which both of them are obnoxious have also been studied~\citep{Gai2022obnoxious}. More generally, the facilities might be of different types and serve a different purpose (for example, they might correspond to a park and a school), in which case the agents might have different, {\em heterogeneous} preferences over them. For example, some agents might be interested in both facilities, some agents might be interested in only one of them, or some agents might consider one facility to be useful and the other to be obnoxious~\citep{anastasiadis2018heterogeneous,chen2020optional,deligkas2023limited,feigenbaum2015hybrid,gai2024mixed,serafino2016,kanellopoulos2021discrete,li2020optional,zou2015dual,Zhao2023constrained,lotfi2023max}. 

In this paper, we consider the case of two different facilities that can be placed only at candidate locations. For each agent, we assume that her position on the line is private, she has approval preferences over the facilities which are publicly known, and her individual cost is given by the total distance from the facilities she approves (rather than the distance from the closest or farthest such facility). We provide more details below.  

\subsection{Our contribution}
We study a truthful two-facility location problem, in which there is a set of agents with {\em known approval preferences} ($0$ or $1$) over two different facilities $\{F_1,F_2\}$, so that each agent approves at least one facility; we can safely ignore agents that approve neither facility. 
The agents have {\em private positions} on the line of real numbers, and the facilities can only be placed at {\em different} locations chosen from a given set of {\em candidate} locations. Once the facilities have both been placed, the individual cost of each agent is the {\em total} distance from the facilities she approves. 

\renewcommand{\arraystretch}{1.3}
\begin{table}[t]
    \centering
    \begin{tabular}{c|cc}
                      & Social cost & Max cost \\ 
       \hline
       Doubleton   & $[1+\sqrt{2},3]$ & $[2,3]$ \\
       Singleton   & $3$  & $3$ \\
       General     & $[3,7]$ & $3$ \\
       \hline
    \end{tabular}
    \caption{An overview of the bounds that we show in this paper on the approximation ratio of deterministic strategyproof mechanisms for the different combinations of social objectives functions (social cost and max cost) and agent preferences (doubleton, singleton, or general).}
    \label{tab:overview}
\end{table}

Our goal is to design mechanisms that take as input the positions reported by the agents, and, using also the available information about the preferences of the agents, decide where to place the two facilities, so that 
(a) a social objective function is (approximately) optimized, and 
(b) the agents are incentivized to truthfully report their positions.   
As in previous work, we consider the well-known {\em social cost} (the total individual cost of the agents) and the {\em max cost} (the maximum individual cost over all agents) as our social objective functions. We treat separately the class of instances in which all agents approve both facilities (to which we refer as {\em doubleton}), the class of instances in which all agents approve one facility (to which we refer as {\em singleton}), and the general class of all possible instances. For all possible combinations of objectives and types of preferences, we design {\em deterministic} strategyproof mechanisms with small, constant approximation ratios. An overview of our results is given in Table~\ref{tab:overview}. 

In Section~\ref{sec:sc} we consider the social cost and show the following results:
\begin{itemize}
\item For doubleton instances (in which all agents approve both facilities), we show that the best possible approximation ratio of strategyproof mechanisms is between $1+\sqrt{2}$ and $3$. Our upper bound follows by a mechanism, which places the facilities at the two candidate locations closest to the median agent; this is the natural extension of the {\sc Median} mechanism which achieves the best possible approximation ratio of $3$ for the single-facility location problem~\citep{feldman2016voting}. These results can be found in Section~\ref{sec:sc-doubleton}.

\item For singleton instances (in which each agent approves one facility), we first observe that no strategyproof mechanism can achieve an approximation ratio better than $3$; this follows from the fact that the problem is now a generalization of the single-facility location problem. The main technical difficulty, which does not allow us to simply treat a singleton instance as two separate single-facility location problems (one for each facility), is that the facilities cannot be placed at the same location. We circumvent this difficulty and show a tight upper bound of $3$ by considering a mechanism that places each facility at the available candidate location closest to the median agent among those that approve it. To decide the order in which the facilities are placed, we first perform a voting step that allows the agents that approve each facility to decide if they prefer the closest or second-closest candidate location to the respective median agent; this is necessary since just blindly choosing the order of placing the facilities leads to a mechanism with a rather large approximation ratio. These results are presented in Section~\ref{sec:sc-singleton}.

\item For general instances, we show an upper bound of $7$ by considering a mechanism which switches between two cases depending on the cardinalities of the sets of agents with different preferences. In particular, when there is a large number of agents that approve both facilities, we run the simple median mechanism we used for doubleton instances by ignoring the other agents. Otherwise, we run a mechanism that places the facility that is approved by most agents at the location closest to the median of the agents that approve {\em only} it, while the other facility is placed at the available location that is closest to the median of the agents that approve it. These results are presented in Section~\ref{sec:sc-general}. Our bound of $7$ for general instances significantly improves upon the bound of $22$ that \citet{lotfi2023max} showed via a reduction between the model in which the individual cost of an agent is the distance to the farthest facility among the ones she approves and our model (in which the individual cost of an agent is the distance to all the facilities she approves).
\end{itemize}

In Section~\ref{sec:max}, we turn our attention to the max cost objective and show the following results:
\begin{itemize}
\item For doubleton instances, we show that the best possible approximation ratio is between $2$ and $3$. Our upper bound follows by a simple mechanism that places the facilities at the available candidate locations closest to the leftmost agent; see Section~\ref{sec:max-doubleton}.

\item For singleton instances, we show a tight bound of $3$ by considering a mechanism that places the two facilities at the candidate locations closest to some agents among those that approve. The main difficulty here is to decide which agents to pick. In particular, after placing the first facility at the candidate location closest to one of the agents that approve it (such as the leftmost), we then need to dynamically decide whether the second facility can be placed closer to the leftmost or rightmost among the agents that approve it, or neither of them. This again is done by a voting-like procedure that is used to decide the order of the agents that approve the second facility relative to the two candidate locations that are closest to where the first facility has been placed. 

\item For general instances, we show a tight bound of $3$ by splitting the class of all instances into those that consist of at least one agent that approves both facilities (in which case we employ the mechanism for doubleton instances) and the remaining instances which are singleton (and we employ the corresponding mechanism). 
\end{itemize}

Finally, in Section~\ref{sec:same}, we consider a slightly simpler model in which the two facilities are allowed to be placed at the same candidate location. For this model, we manage to show improved, tight bounds on the approximation ratio of deterministic mechanisms for doubleton and general instances for both the social and the max cost (the problem is not interesting for singleton instances). This is possible because we can now avoid possible misreports by agents with doubleton preferences, which in turn allows us to consider a class of mechanisms that is not strategyproof when the facilities are constrained to be placed at different locations.

\subsection{Related work}
For an overview of the many different truthful facility location problems that have been considered in the literature, we refer the reader to the survey of \citet{fl-survey}. Here, we will briefly discuss and compare the papers that are most related to our work. Most of the papers discussed below differ from ours in at least one modeling dimension in terms of the definition of the individual cost of the agents, the possible constraints on the locations of the facilities, and what type of information related to the positions and preferences of the agents is assumed to be private or public. 

The truthful two-facility location problem was first considered in the original work of \citet{procaccia09approximate}, in which the goal is to locate two identical facilities (even at the same location) and the individual cost of an agent is the distance between her position and the closest facility. 
For deterministic mechanisms, \citeauthor{procaccia09approximate} showed a constant lower bound and a linear upper bound on the approximation ratio in terms of the social cost, and a tight bound of $3$ in terms of the maximum cost. They also showed how randomization can lead to further improvements. 
\citet{Lu2010two-facility} improved the lower bound for the social cost and deterministic mechanisms to an asymptotically linear one, before \citet{fotakis2016concave} finally showed that the exact bound for this case is $n-2$. 

\citet{Sui2015constrained} were among the first to consider truthful facility location problems with candidate locations (referred to as constrained facility location), with a focus on achieving approximate strategyproofness by bounding the incentives of the agents to manipulate; for multiple facilities, they considered only doubleton instances where each agent's individual cost is the distance to the closest facility.  
As already mentioned, \citet{feldman2016voting} considered a candidate selection problem with a fixed set of candidates, a model which translates into a single-facility location problem where the facility can only be placed at a location from a given set of discrete candidate locations. They focused on the social cost objective and, among other results, proved that the {\sc Median} mechanism that places the facility at the location closest to the position reported by the median agent, achieves an upper bound of $3$; they also showed that this is the best possible bound among deterministic mechanisms.

\citet{serafino2016} considered a slightly different discrete facility location problem, where agents occupy nodes on a line graph and have approval preferences over two different facilities that can only be placed at different nodes of the line. In contrast to our work here, where we assume that the positions of the agents are private and their preferences public information, \citeauthor{serafino2016} assumed that the positions are known and the preferences unknown. They showed several bounds on the approximation ratio for deterministic and randomized strategyproof mechanisms for the social cost and the max cost. Some of their results for deterministic mechanisms were improved by \citet{kanellopoulos2021discrete}. The alternative continuous model (with the same assumptions about the positions and preferences of the agents as \citet{serafino2016} ) where the agents have positions on the line of real numbers and the two facilities can be located at any point of the line was considered by \citet{chen2020optional} and in the follow-up work of \citet{li2020optional}. \citeauthor{chen2020optional} showed bounds on the approximation ratio of strategyproof mechanisms for two different individual cost definitions depending on whether the cost of an agent is determined by closest or the farthest approved facility; \citet{li2020optional} showed improved bounds for the former individual cost definition.

\citet{Tang2020candidate} considered a setting in which two identical facilities can be placed at locations chosen from a set of candidate ones, allowing the facilities to be placed even at the same location. The positions of the agents are assumed to be private information and an agent's individual cost is defined as her distance from the closest facility. They proved an upper bound of $2n-3$ for the social cost objective and a tight bound of $3$ for the maximum cost. They also considered the case of a single facility and the max cost objective, for which they showed a bound of $3$ (extending the work of \citet{feldman2016voting} who only focused on the social cost). 
\citet{Walsh2021limited} considered a similar setting, where one or more facilities can only be placed at different subintervals of the line, and showed bounds on the approximation ratio of strategyproof mechanisms for many social objective functions, beyond the classic ones.
\citet{Zhao2023constrained} studied a slightly different setting, in which the agents have known approval preferences over two different facilities and their individual costs are defined as their distance from the farthest facility among the ones they approve. For doubleton instances, they showed a tight bound of $3$ for both the social cost and the max cost objectives, while for general instances they showed an upper bound of $2n+1$ for the social cost and an upper bound of $9$ for the max cost; their results for general instances were recently improved by \citet{lotfi2023max} to $11$ and $5$ for the social cost and the max cost, respectively. As already mentioned \citet{lotfi2023max} also showed a bound of $22$ via a reduction between their individual max cost model and the individual sum cost model that we focus here; we improve this bound to $7$.  

\citet{Xu2021minimum} considered a setting where two facilities must be located so that there is a minimum distance between them (not at specific given candidate locations). They showed results for two types of individual costs. The first one is, as in our case, the total distance (assuming that the facilities play a different role, and thus the agents are interested in both of them) and showed that, for any minimum distance requirement, the optimal solution for the social cost or the maximum cost can be attained by a strategyproof mechanism. The second one is the minimum distance (assuming that the facilities are of the same type, and thus the agents are interested only in their closest one), and showed that the approximation ratio of strategyproof mechanisms is unbounded. They also considered the case where the facility is obnoxious and showed a bound that depends on the minimum distance parameter. \citet{Duan2021minimum} later generalized the minimum distance setting by allowing for private fractional preferences over the two facilities.


\section{Preliminaries}
We consider the two-facility location problem with candidate locations. 
An instance $I$ of this problem consists of a set $N$ of $n \geq 2$ agents and two facilities $\{F_1, F_2\}$.
Each agent $i \in N$ has a {\em private position} $x_i \in \mathbb{R}$ on the line of real numbers, and a {\em known approval preference} $p_{ij} \in \{0,1\}$ for each $j \in [2]$, indicating whether she approves facility $F_j$ ($p_{ij}=1$) or not ($p_{ij}=0$), such that $p_{i1}+p_{i2}\geq 1$. There is also a set of $m \geq 2$ candidate locations $C$ where the facilities can be located. To be concise, we denote an instance using the tuple $I=(\bx,\bp,C)$, where $\bx = (x_i)_{i \in N}$ is the {\em position profile} of all agent positions, and $\bp = (p_{ij})_{i \in N, j \in [2]}$ is the {\em preference profile} of all agent approval preferences.

A {\em feasible solution} (or, simply, {\em solution}) is a pair $\bc =(c_1,c_2) \in C^2$ of candidate locations with $c_1 \neq c_2$, where the two facilities can be placed; that is, for each $j\in [2]$, $F_j$ is placed at $c_j$. 
A {\em mechanism} $M$ takes as input an instance $I$ of the problem and outputs a feasible solution $M(I)$. 
Our goal is to design mechanisms so that (a) some social objective function is (approximately) optimized, and (b) the agents truthfully report their private positions. 

The {\em individual cost} of an agent $i \in N$ for a solution $\bc$ is her total distance from the locations of the facilities she approves:
$$\cost_i(\bc|I) = \sum_{j \in [2]} p_{ij} \cdot d(x_i,c_j),$$
where $d(x,y) = |x-y|$ denotes the distance between any two points $x$ and $y$ on the line. Since the line is a special metric space, the distances satisfy the {\em triangle inequality}, which states that $d(x,y) \leq d(x,z) + d(z,y)$ for any three points $x$, $y$ and $z$ on the line, with the equality being true when $z \in [x,y]$. 
We consider the following two natural social objective functions that have been considered extensively within the truthful facility location literature:
\begin{itemize}
\item The {\em social cost} of a solution $\bc$ is the total individual cost of the agents:
$$\SC(c|I) = \sum_{i \in N} \cost_i(\bc).$$ 
\item The {\em max cost} of a solution $\bc$ is the maximum individual cost over all agents:
$$\MC(c|I) = \max_{i \in N} \cost_i(\bc).$$ 
\end{itemize}
The {\em approximation ratio} of a mechanism $M$ in terms of a social objective function $f \in \{\SC,\MC\}$ is the worst-case ratio (over all possible instances) of the $f$-value of the solution computed by the mechanism over the minimum possible $f$-value over all possible solutions:
$$\sup_{I} \frac{f(M(I)|I)}{\min_{\bc \in C^2}f(c|I)}.$$

A mechanism is said to be {\em strategyproof} if the solution $M(I)$ it returns when given as input any instance $I=(\bx,\bp,C)$ is such that there is no agent $i$ with incentive to misreport a position $x_i' \neq x_i$ to decrease her individual cost, that is, 
\begin{align*}
\cost_i(M(I)|I) \leq \cost_i(M((x_i',\bx_{-i}),\bp,C)|I),
\end{align*}
where $(x_i',\bx_{-i})$ is the position profile obtained by $\bx$ when only agent $i$ reports a different position $x_i'$.

Finally, let us introduce some further notation and terminology that will be useful. 
For each $j \in [2]$, we denote by $N_j$ the set of agents that approve facility $F_j$, i.e.,  $i \in N_j$ if $p_{ij}=1$. Any agent that approves both facilities belongs to the intersection $N_1 \cap N_2$ and has a {\em doubleton preference}. Any agent that approves one facility belongs to either $N_1 \setminus N_2$ or $N_2 \setminus N_1$ and has a {\em singleton preference}.  
Besides {\em general} instances (with agents that have any type of approval preferences), we will also pay particular attention to the following two classes of instances: 
\begin{itemize}
    \item {\em Doubleton:} All agents have a doubleton preference, that is, $N_1 \cap N_2=N$; 
    \item {\em Singleton:} All agents have a singleton preference, that is, $N_1 \cap N_2 = \varnothing$.
\end{itemize}
We will also denote by $m_j$, $\ell_j$, and $r_j$ the median\footnote{Without loss of generality, we break potential  ties in favor of the leftmost median agent.}, leftmost, and rightmost, respectively, agent in $N_j$. In addition, for any agent $i$ we denote by $t(i)$ and $s(i)$ the closest and the second closest, respectively, candidate location to $i$.


\section{Social cost} \label{sec:sc}
In this section we will focus on the social cost. We will show that the best possible approximation ratio of strategyproof mechanisms is between $1+\sqrt{2}$ and $3$ for doubleton instances, exactly $3$ for singleton instances, and between $3$ and $7$ for general instances.  

\subsection{Doubleton instances} \label{sec:sc-doubleton}
We start with the case of doubleton instances in which all agents approve both facilities. Recall that for the single-facility location problem, \citet{feldman2016voting} showed that the best possible approximation ratio of $3$ is achieved by the {\sc Median} mechanism, which places the facility at the candidate location closest to the position reported by the median agent $m$. We can generalize this mechanism by placing the two facilities at the two candidate locations that are closest to the position reported by $m$; that is, $F_1$ is placed at $w_1=t(m)$ and $F_2$ is placed at $w_2=s(m)$; see Mechanism~\ref{mech:median}. It is not hard to show that this is a strategyproof mechanism; the median agent minimizes her cost and any other agent would have to become the median agent to manipulate the outcome which could only lead to placing the facilities farther away. We next show that the mechanism achieves an approximation ratio of at most $3$, but cannot do better; we remark that this result has also been independently shown by \citet{gai2024mixed} when all agents are of type-II in their model. 

\newcommand\mycommfont[1]{\normalfont\textcolor{blue}{#1}}
\SetCommentSty{mycommfont}
\begin{algorithm}[h]
\SetNoFillComment
\caption{\sc Median}
\label{mech:median}
{\bf Input:} Reported positions of agents with doubleton preferences\;
{\bf Output:} Facility locations $\bw = (w_1,w_2)$\;
$m \gets$ median agent in $N_1 \cap N_2$\;
$w_1 \gets t(m)$\;
$w_2 \gets s(m)$\;
\end{algorithm}

\begin{theorem}\label{thm:sc-median-3}
For doubleton instances, the approximation ratio of the {\sc Median} mechanism is at most $3$, and this is tight. 
\end{theorem}

\begin{proof}
Let $\opt = (o_1,o_2)$ be an optimal solution. 
Since the position of the median agent minimizes the total distance of all agents, we have that  
$$\sum_{i \in N} d(i,m) \leq \sum_{i \in N} d(i,x)$$ 
for any point $x$ of the line (including $o_1$ and $o_2$), and thus
$$2\sum_{i \in N} d(i,m) \leq \sum_{i \in N} d(i,o_1) + \sum_{i \in N} d(i,o_2) = \SC(\opt).$$ 
Also, since $t(m)$ and $s(m)$ are the two closest candidate locations to $m$, we have that $d(m,t(m)) \leq d(m,x)$ for any candidate location $x$, and there exists $o \in \{o_1,o_2\}$ such that $d(m,s(m)) \leq d(m,o)$; let $\tilde{o} \in \{o_1,o_2\}\setminus \{o\}$. 
Therefore, using these facts and the triangle inequality, we obtain
\begin{align*}
\SC(\bw) &= \sum_{i \in N} \bigg( d(i,t(m)) + d(i,s(m)) \bigg) \\
&\leq 2\sum_{i \in N} d(i,m) + \sum_{i \in N} d(m,t(m)) + \sum_{i \in N} d(m,s(m)) \\
&\leq \SC(\opt) + \sum_{i \in N} d(m,\tilde{o}) + \sum_{i \in N} d(m,o) \\
&\leq \SC(\opt) + 2\sum_{i \in N} d(i,m) + \sum_{i \in N}d(i,\tilde{o}) + \sum_{i \in N} d(i,o) \\
&\leq 3\cdot \SC(\opt).
\end{align*}

The analysis of the mechanism is tight due to the following instance: 
There are four candidate locations at $0$, $\varepsilon$, $1-\varepsilon$, and $1$, for some infinitesimal $\varepsilon > 0$. 
There are also two agents positioned at $1/2-\varepsilon$ and $1$, respectively. 
Let the first agent be the median one (in case the second agent is the median, there is a symmetric instance). 
Then, the two facilities are placed at $0$ and $\varepsilon$ for a social cost of approximately $3$, whereas the optimal solution is to place the facilities at $1-\varepsilon$ and $1$ for a social cost of approximately $1$, leading to a lower bound of nearly $3$. 
\end{proof}

We next show a lower bound of $1+\sqrt{2}$ on the approximation ratio of any strategyproof mechanism.  

\begin{theorem} \label{thm:sc-doubleton-lower}
For doubleton instances, the approximation ratio of any strategyproof mechanism is at least $1+\sqrt{2}-\delta$, for any $\delta > 0$. 
\end{theorem}

\begin{proof}
Let $\varepsilon > 0$ be an infinitesimal. 
We will consider instances with four candidate locations, two in the $\varepsilon$-neighborhood of $0$ (for example, $-\varepsilon$ and $\varepsilon$) and two in the $\varepsilon$-neighborhood of $2$ (for example, $2-\varepsilon$ and $2+\varepsilon)$. To simplify the calculations in the remainder of the proof, we will assume that there can be candidate locations at the same point of the line, so that we have two candidate locations at $0$ and two at $2$. 

First, consider the following generic instance $I$ with the aforementioned candidate locations: There is at least one agent at $0$, at least one agent at $2$, while each remaining agent is arbitrarily located at a location from $\{0,1-\varepsilon,1+\varepsilon, 2\}$.
We make the following observation: 
Any solution returned by a strategyproof mechanism when given as input $I$ must also be returned when given as input any of the following two instances:
\begin{itemize}
    \item $J_1$: Same as $I$ with the difference that an agent $j_1$ has been moved from $0$ to $1-\varepsilon$. 
    \item $J_2$: Same as $I$ with the difference that an agent $j_2$ has been moved from $2$ to $1+\varepsilon$. 
\end{itemize}
Suppose towards a contradiction that this is not true for $J_1$; similar arguments can be used for $J_2$. We consider the following cases:
\begin{itemize}
    \item Both facilities are placed at $2$ in $I$. If this is not done in $J_1$, then $j_1$ can misreport her position as $1-\varepsilon$ in $I$ so that the instance becomes $J_1$ and at least one facility moves to her true position $0$.
    
    \item Both facilities are placed at $0$ in $I$. If this is not done in $J_1$, then $j_1$ can misreport her position as $0$ in $J_1$ so that the instance becomes $I$ and both facilities move to $0$ which is closer to her true position $1-\varepsilon$, a contradiction.
    
    \item One facility is placed at $0$ and the other is placed at $2$ in $I$.
    Observe that it cannot be the case that both facilities are placed at $0$ in $J_1$ since that would mean that $j_1$ can misreport her position in $I$ as $1-\varepsilon$ so that the instance becomes $J_1$ and both facilities move to her true position $0$. So, the only possibility of having a different solution in $I$ and $J_1$ is that both facilities are placed at $2$ in $J_1$. But then, $j_1$ can misreport her true position as $0$ in $J_1$ so that the instance becomes $I$ and one of the facilities moves to $0$ which is closer to her true position. 
\end{itemize}
Hence, the same solution must be computed by the mechanism when given $I$ or $J_1$ as input.

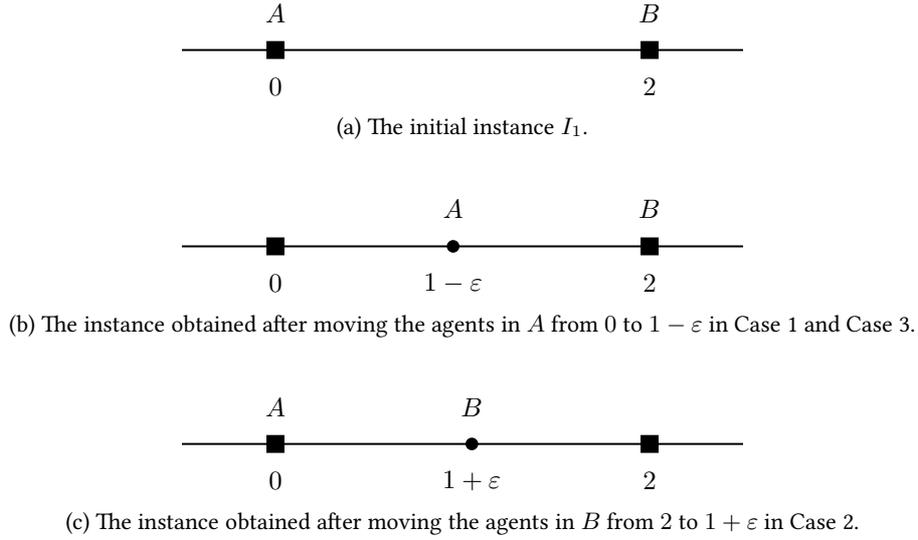
\begin{figure}[t]
\tikzset{every picture/.style={line width=0.75pt}} 
\centering
\begin{subfigure}[t]{\linewidth}
\centering
\begin{tikzpicture}[x=0.7pt,y=0.7pt,yscale=-1,xscale=1]
\draw [line width=0.75]  (0,0) -- (300,0) ;
\filldraw ([xshift=-3pt,yshift=-3pt]50,0) rectangle ++(6pt,6pt);
\filldraw ([xshift=-3pt,yshift=-3pt]250,0) rectangle ++(6pt,6pt);

\draw (50,-20) node [inner sep=0.75pt]  [font=\small]  {$A$};
\draw (250,-20) node [inner sep=0.75pt]  [font=\small]  {$B$};

\draw (50,20) node [inner sep=0.75pt]  [font=\small]  {$0$};
\draw (250,20) node [inner sep=0.75pt]  [font=\small]  {$2$};

\end{tikzpicture}
\caption{The initial instance $I_1$.}
\label{fig:sc-doubleton-lower-a}
\end{subfigure}
\\[20pt]
\begin{subfigure}[t]{\linewidth}
\centering
\begin{tikzpicture}[x=0.7pt,y=0.7pt,yscale=-1,xscale=1]
\draw [line width=0.75]  (0,0) -- (300,0) ;
\filldraw ([xshift=-3pt,yshift=-3pt]50,0) rectangle ++(6pt,6pt);
\filldraw ([xshift=-3pt,yshift=-3pt]250,0) rectangle ++(6pt,6pt);
\filldraw (145,0) circle (2pt);

\draw (145,-20) node [inner sep=0.75pt]  [font=\small]  {$A$};
\draw (250,-20) node [inner sep=0.75pt]  [font=\small]  {$B$};

\draw (50,20) node [inner sep=0.75pt]  [font=\small]  {$0$};
\draw (250,20) node [inner sep=0.75pt]  [font=\small]  {$2$};

\draw (145,20) node [inner sep=0.75pt]  [font=\small]  {$1-\varepsilon$};

\end{tikzpicture}
\caption{The instance obtained after moving the agents in $A$ from $0$ to $1-\varepsilon$ in Case 1 and Case 3.}
\label{fig:sc-doubleton-lower-b}
\end{subfigure}
\\[20pt]
\begin{subfigure}[t]{\linewidth}
\centering
\begin{tikzpicture}[x=0.7pt,y=0.7pt,yscale=-1,xscale=1]
\draw [line width=0.75]  (0,0) -- (300,0) ;
\filldraw ([xshift=-3pt,yshift=-3pt]50,0) rectangle ++(6pt,6pt);
\filldraw ([xshift=-3pt,yshift=-3pt]250,0) rectangle ++(6pt,6pt);
\filldraw (155,0) circle (2pt);

\draw (50,-20) node [inner sep=0.75pt]  [font=\small]  {$A$};
\draw (155,-20) node [inner sep=0.75pt]  [font=\small]  {$B$};

\draw (50,20) node [inner sep=0.75pt]  [font=\small]  {$0$};
\draw (250,20) node [inner sep=0.75pt]  [font=\small]  {$2$};

\draw (155,20) node [inner sep=0.75pt]  [font=\small]  {$1+\varepsilon$};
\end{tikzpicture}
\caption{The instance obtained after moving the agents in $B$ from $2$ to $1+\varepsilon$ in Case 2.}
\label{fig:sc-doubleton-lower-c}
\end{subfigure}
\caption{The instances used in the proof of the lower bound of $1+\sqrt{2}$ in terms of the social cost for doubleton instances (Theorem~\ref{thm:sc-doubleton-lower}). Set $A$ consists of $\alpha n$ agents and set $B$ consists of $(1-\alpha)n$ agents; all of them approve both facilities. Rectangles represent candidate locations; recall that we assume that there are two candidate locations arbitrarily close to $0$ and two candidate locations arbitrarily close to $2$.}
\label{fig:sc-doubleton-lower}
\end{figure}

Now, consider an arbitrary strategyproof mechanism and let $\alpha = \sqrt{2}-1$; note that $\alpha$ is such that $\frac{1+\alpha}{1-\alpha} = \frac{1}{\alpha}=1+\sqrt{2}$. Let $I_1$ be the following instance with the aforementioned candidate locations: $\alpha n$ agents are at $0$ and $(1-\alpha)n$ agents are at $2$. See Figure~\ref{fig:sc-doubleton-lower-a}. We consider the following cases depending on the solution returned by the mechanism when given $I_1$ as input:

\medskip
\noindent
{\bf Case 1:} The mechanism places both facilities at $0$. 
We consider the sequence of instances obtained by moving one by one the $\alpha n$ agents that are positioned at $0$ in $I_1$ to $1-\varepsilon$; see Figure~\ref{fig:sc-doubleton-lower-b}. By the observation above, the mechanism must return the same solution for any two consecutive instances of this sequence (essentially, the first one is of type $I$ and the second one is of type $J_1$), which means that the mechanism must eventually return the same solution for all of them. Therefore, the mechanism must place both facilities at $0$ in the last instance of this sequence, where $\alpha n$ agents are at $1-\varepsilon$ and the remaining $(1-\alpha)n$ agents are at $2$. This solution has social cost $2\alpha n + 4(1-\alpha)n = 2(2-\alpha)n$. However, the solution that places both facilities at $2$ has social cost $2\alpha n$, leading to an approximation ratio of  $\frac{2}{\alpha}-1 > 1+\sqrt{2}$. 

\medskip
\noindent
{\bf Case 2:} The mechanism places both facilities at $2$. 
Similarly to Case 1 above, we now consider the sequence of instances obtained by moving one by one the $(1-\alpha) n$ agents that are positioned at $2$ in $I_1$ to $1+\varepsilon$; see Figure~\ref{fig:sc-doubleton-lower-c}. Again, by the observation above, the mechanism must return the same solution for any two consecutive instances of this sequence (the first one is of type $I$ and the second one is of type $J_2$), which means that the mechanism must eventually return the same solution for all of them. 
Therefore, the mechanism must place both facilities at $2$ in the last instance of this sequence, where $\alpha n$ agents are at $0$ and the remaining $(1-\alpha)n$ agents are at $1+\varepsilon$. 
This solution has social cost $4\alpha n + 2(1-\alpha)n = 2(1+\alpha)n$. 
However, the solution that places both facilities at $0$ has social cost $2(1-\alpha) n$, leading to an approximation ratio of  $\frac{1+\alpha}{1-\alpha}=1+\sqrt{2}$.

\medskip
\noindent
{\bf Case 3:} The mechanism places one facility at $0$ and the other at $2$.
We consider the same sequence of instances as in Case 1. This results in that the mechanism must place one facility at $0$ and the other at $2$ when given as input the instance where $\alpha n$ agents are at $1-\varepsilon$ while the remaining $(1-\alpha)n$ agents are at $2$. This solution has social cost $2\alpha n + 2(1-\alpha)n = 2n$. However, the solution that places both facilities at $2$ has social cost $2\alpha n$, leading to an approximation ratio of  $\frac{1}{\alpha} = 1+\sqrt{2}$. 
\end{proof}


\subsection{Singleton instances} \label{sec:sc-singleton}
It is not hard to observe that our two-facility problem with singleton instances is more general than the single-facility location problem studied by \citet{feldman2016voting}; indeed, there are singleton instances in which all agents approve the same facility, and thus the location of the other facility does not affect the social cost nor the approximation ratio. Consequently, we cannot hope to achieve an approximation ratio better than $3$. For completeness, we include here a slightly different proof of the lower bound of $3$ for all strategyproof mechanisms with instances that involve agents that approve different facilities. Recall that, for singleton instances, $N_1 \cap N_2 = \varnothing$.

\begin{theorem} \label{thm:sc-general-lower}
For singleton instances, the approximation ratio of any strategyproof mechanism is at least $3-\delta$, for any $\delta > 0$. 
\end{theorem}

\begin{proof}
Let $\varepsilon >0$ be an infinitesimal and consider an instance $I_1$ with two candidate locations at $-1$ and $1$, and two agents positioned at $\varepsilon > 0$ such that one of them approves $F_1$ while the other approves $F_2$; see Figure~\ref{fig:sc-general-lower-a}.
There are two possible solutions, $(-1,1)$ or $(1,-1)$.  
Without loss of generality, suppose that $(1,-1)$ is the solution chosen by an arbitrary strategyproof mechanism.

\begin{figure}[t]
\tikzset{every picture/.style={line width=0.75pt}} 
\centering
\begin{subfigure}[t]{0.45\linewidth}
\centering
\begin{tikzpicture}[x=0.7pt,y=0.7pt,yscale=-1,xscale=1]
\draw [line width=0.75]  (0,0) -- (300,0) ;
\filldraw ([xshift=-3pt,yshift=-3pt]50,0) rectangle ++(6pt,6pt);
\filldraw ([xshift=-3pt,yshift=-3pt]250,0) rectangle ++(6pt,6pt);
\filldraw (155,0) circle (2pt);

\draw (155,-20) node [inner sep=0.75pt]  [font=\small]  {$i,j$};

\draw (50,20) node [inner sep=0.75pt]  [font=\small]  {$-1$};
\draw (250,20) node [inner sep=0.75pt]  [font=\small]  {$1$};
\draw (155,20) node [inner sep=0.75pt]  [font=\small]  {$\varepsilon$};

\end{tikzpicture}
\caption{Instance $I_1$.}
\label{fig:sc-general-lower-a}
\end{subfigure}
\\[20pt]
\begin{subfigure}[t]{0.45\linewidth}
\centering
\begin{tikzpicture}[x=0.7pt,y=0.7pt,yscale=-1,xscale=1]
\draw [line width=0.75]  (0,0) -- (300,0) ;
\filldraw ([xshift=-3pt,yshift=-3pt]50,0) rectangle ++(6pt,6pt);
\filldraw ([xshift=-3pt,yshift=-3pt]250,0) rectangle ++(6pt,6pt);
\filldraw (155,0) circle (2pt);

\draw (155,-20) node [inner sep=0.75pt]  [font=\small]  {$i$};
\draw (250,-20) node [inner sep=0.75pt]  [font=\small]  {$j$};

\draw (50,20) node [inner sep=0.75pt]  [font=\small]  {$-1$};
\draw (250,20) node [inner sep=0.75pt]  [font=\small]  {$1$};
\draw (155,20) node [inner sep=0.75pt]  [font=\small]  {$\varepsilon$};

\end{tikzpicture}
\caption{Instance $I_2$.}
\label{fig:sc-general-lower-b}
\end{subfigure}
\caption{The two instances used in the proof of the lower bound of $3$ in terms of the social cost for the general case (Theorem~\ref{thm:sc-general-lower}). Agent $i$ approves $F_1$ and agent $j$ approves $F_2$. Rectangles represent candidate locations.}
\label{fig:sc-general-lower}
\end{figure}
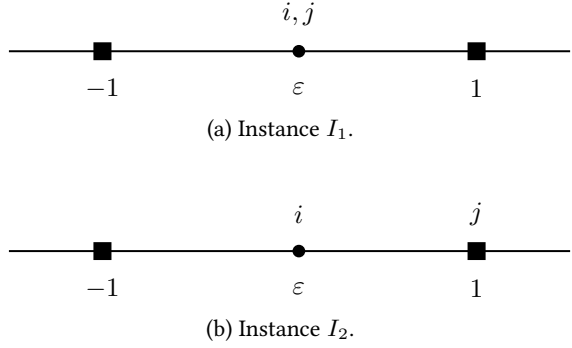

Next, consider instance $I_2$, which is the same as $I_1$, with the only difference that the agent that approves $F_2$ is moved from $\varepsilon$ to $1$; see Figure~\ref{fig:sc-general-lower-b}. To maintain strategyproofness, the solution $(1,-1)$ must be returned in $I_2$ as well; otherwise, the moving agent would have decreased her cost in $I_1$ from $1+\varepsilon$ to $1-\varepsilon$. This solution has social cost $1-\varepsilon+2=3-\varepsilon$, whereas the other solution $(-1,1)$ has social cost just $1+\varepsilon$, leading to a lower bound of $3-\delta$, for any $\delta > 0$.
\end{proof}

Since there is an adaptation of the {\sc Median} mechanism that achieves an approximation ratio of $3$ for doubleton instances (see Theorem~\ref{thm:sc-median-3}), one might wonder if there is a variant that can do so for singleton instances as well. In particular, the natural extension of {\sc Median} is to place $F_1$ at the candidate location closest to the (leftmost) median agent $m_1$ of $N_1$, and $F_2$ at the available candidate location closest to the (leftmost) median agent $m_2$ of $N_2$. While this seems like a good idea at first glance, the following example shows that it fails to achieve the desired approximation ratio bound. 

\begin{example}
Consider an instance with two candidate locations at $0$ and $2$. For some $x \geq 1$, there are $2x+1$ agents that approve only $F_1$ such that $x+1$ of them are located at $1-\varepsilon$ and the other $x$ are located at $2$. There are also $2x+1$ agents that approve only $F_2$ and are all located at $0$. According to the definition of the mechanism, $F_1$ is placed at $0$ (which is the candidate location closest to the median agent in $N_1$), and then $F_2$ is placed at $2$ as $0$ is now occupied and $2$ is available. This solution has social cost approximately $(x+2x) + 4x = 7x$, whereas the solution that places $F_1$ at $2$ and $F_2$ at $0$ has social cost approximately $x$, leading to an approximation ratio of nearly $7$. 
\end{example}

The issue with the aforementioned variant of the {\sc Median} mechanism is the order in which it decides to place the facilities. If it were to place $F_2$ first and $F_1$ second then it would have made the optimal choice in the example. However, there is a symmetric example that would again lead to a lower bound of approximately $7$. So, the mechanism needs to be able to dynamically determine the order in which it places $F_1$ and $F_2$. This brings us to the following idea: 
We will again place each facility one after the other at the closest candidate location to the median among the agents that approve it. However, the facility that is placed first (and thus has priority in case the median agents of $N_1$ and $N_2$ are closer to the same candidate location) is the one with stronger majority in terms of the number of agents that approve it who are closer to the top choice of the median agent rather than her second choice; ties are broken in favor of the facility that is approved by most agents, which is assumed to be $F_1$ without loss of generality. We refer to this mechanism as $\PCM$; see Mechanism~\ref{mech:sc-PCM} for a more formal description. 

\SetCommentSty{mycommfont}
\begin{algorithm}[h]
\SetNoFillComment
\caption{$\PCM$}
\label{mech:sc-PCM}
{\bf Input:} Reported positions of agents with singleton preferences\;
{\bf Output:} Facility locations $\bw = (w_1,w_2)$\;
\For{$j \in [2]$}{
$m_j \gets$ median agent in $N_j$\;
$S_j \gets$ set of agents in $N_j$ (weakly) closer to $t(m_j)$ than to $s(m_j)$\;
}
\uIf{$2|S_1|-|N_1| \geq 2|S_2|-|N_2|$}{ 
    $j \gets 1$\;
}
\Else{
    $j \gets 2$\;
}

$w_j \gets t(m_j)$\;
\uIf{$t(m_{3-j})$ is available}{
    $w_{3-j} \gets t(m_{3-j})$\;
}
\Else{
    $w_{3-j} \gets s(m_{3-j})$\;
}
\end{algorithm}

We first show that this mechanism is strategyproof; it is not hard to observe that this must be true as the mechanism is a composition of variants of two simple strategyproof mechanisms (median plus majority voting). 

\begin{theorem}
$\PCM$ is strategyproof.  
\end{theorem}

\begin{proof}
Clearly, if $t(m_1) \neq t(m_2)$ then no agent has incentive to deviate as then the facilities are placed at $t(m_1)$ and $t(m_2)$ independently of whether $2|S_j|-|N_j| \geq 2|S_{3-j}|-|N_{3-j}|$  or not for $j\in [2]$. 
So, it suffices to consider the case where $t(m_1) = t(m_2)$ and $2|S_j|-|N_j| \geq 2|S_{3-j}|-|N_{3-j}|$ for some $j \in [2]$, leading to $w_j = t(m_j)$ and $w_{3-j} = s(m_{3-j})$, solving ties in favor of $F_1$. 
\begin{itemize}
    \item $m_j$ and any agent $i \in N_j$ that is closer to $t(m_j)$ than to $s(m_j)$ have no incentive to deviate as $t(m_j)$ is the best choice for them. 
    \item Any agent $i \in N_j$ that is closer to $s(m_j)$ than to $t(m_j)$ has no incentive to deviate, as going closer to $t(m_j)$ can only increase the quantity $2|S_j|-|N_j|$ and cannot change the outcome.
    \item $m_{3-j}$ and any agent $i \in N_{3-j}$ that is closer to $t(m_{3-j})$ than to $s(m_{3-j})$ have no incentive to deviate as moving closer to $s(m_{3-j})$ would decrease the quantity $2|S_{3-j}|-|N_{3-j}|$ and would not change the outcome.
    \item Any agent $i \in N_{3-j}$ that is closer to $s(m_{3-j})$ than to $t(m_{3-j})$ has no incentive to deviate as $s(m_{3-j})$ is the best choice for her. 
\end{itemize}
So, the mechanism is strategyproof. 
\end{proof}

Next, we show the upper bound of $3$ on the approximation ratio.

\begin{theorem} \label{thm:sc-singleton-upper}
The approximation ratio of $\PCM$ is at most $3$. 
\end{theorem}

\begin{proof}
Let $\opt = (o_1,o_2)$ be an optimal solution; without loss of generality, we can assume that $w_1 < w_2$ and $o_1 < o_2$. 
We consider the following two cases:

\medskip
\noindent
{\bf Case 1: $t(m_1) \neq t(m_2)$.}
Then, we have that $w_1 = t(m_1)$ and $w_2 = t(m_2)$. By the properties of the median, for any $j \in [2]$, we have that 
$$\sum_{i \in N_j} d(i,m_j) \leq \sum_{i \in N_j} d(i,x)$$ 
for any point $x$ of the line, including $o_j$. Also, by the definition of $t(m_j)$, we have that
$d(m_j,t(m_j)) \leq d(m_j,x)$ for any candidate location $x$, again including $o_j$. 
Therefore, using these facts and the triangle inequality, we obtain
\begin{align*}
\SC(\bw) &= \sum_{j \in [2]} \sum_{i \in N_j} d(i,t(m_j)) \\
&\leq \sum_{j \in [2]} \sum_{i \in N_j} d(i,m_j) + \sum_{j \in [2]} \sum_{i \in N_j} d(m_j,t(m_j)) \\
&\leq \sum_{j \in [2]} \sum_{i \in N_j} d(i,m_j) +  \sum_{j \in [2]} \sum_{i \in N_j} d(m_j,o_j)\\
&\leq 2 \cdot \sum_{j \in [2]} \sum_{i \in N_j} d(i,m_j) +  \sum_{j \in [2]} \sum_{i \in N_j} d(i,o_j) \\
&\leq 3 \cdot \sum_{j \in [2]} \sum_{i \in N_j} d(i,o_j) \\
&= 3 \cdot \SC(\opt). 
\end{align*}

\medskip
\noindent
{\bf Case 2: $t(m_1) = t(m_2)$.} 
We can without loss of generality focus on the case where $2|S_1|-|N_1| \geq 2|S_2|-|N_2|$; the case where the inequality is the other way around can be handled using similar arguments. 
So, $w_1=t(m_1)$ and $w_2 = s(m_2)$. Note that $|S_1| \geq |N_1|/2$ and $|S_2| \geq |N_2|/2$. 
If $m_2$ is closer to $s(m_2)$ than to $o_2$, then we can repeat the arguments of Case 1 to obtain an upper bound of $3$. So, it suffices to focus on the case where $m_2$ is closer to $o_2$ than to $s(m_2)$, which means that $o_2 = t(m_2)$, and thus $o_1 < w_1=o_2 < w_2$. 
Now, observe the following:
\begin{itemize}
    
    \item Since $m_1$ is closer to $w_1$ than to $o_1$, we can move the agents in $S_1$ at $\frac{o_1+w_1}{2}$ and the remaining $|N_1|-|S_1|$ agents at $o_1$. Doing this, the approximation ratio cannot decrease (as we move towards $o_1$ and either towards $w_1$ at the same rate or away from $w_1$), $m_1$ remains the median agent of $N_1$ since $|S_1| \geq |N_1|/2$, and it is still true that $t(m_1) = w_1$. 

    \item We have that $m_2$ is closer to $o_2$ than to $w_2$.
    If $m_2 \leq o_2$, then we can move the agents of $N_2$ as follows: each agent in $S_2$ is moved at $o_2$ and the remaining $|N_2|-|S_2|$ agents of $N_2$ (who are closer to $w_2$ than to $o_2$) at $\frac{o_2+w_2}{2}$. 
    Doing this, the approximation ratio cannot decrease, $m_2$ remains the median agent of $N_2$ as $|S_2|\geq |N_2|/2$, and clearly, it is still true that $s(m_2)=w_2$. 
    
    If $m_2 > o_2$, then we can move the agents of $N_2$ as follows: 
    $|N_2|/2$ agents at $o_2$,  
    $|S_2|-|N_2|/2$ agents at $m_2$, and 
    the remaining $|N_2|-|S_2|$ agents (who are closer to $w_2$ than to $o_2$) at $\frac{o_2+w_2}{2}$.
    Doing this, the approximation ratio cannot decrease, $m_2$ remains the median agent of $N_2$ as $|S_2|\geq |N_2|/2$, and clearly, it is still true that $s(m_2)=w_2$. 

    It is not hard to observe that the first case ($m_2 \leq o_2$) is worse in terms of approximation ratio than the second case ($m_2 > o_2$) as more agents are exactly at their optimal location. So, it suffices to consider this one. 
\end{itemize}
Based on the above, in the worst case, we have
\begin{align*}
\SC(\bw) 
&= (|N_1|-|S_1|)(w_1-o_1) + |S_1|\frac{w_1-o_1}{2} + |S_2|(w_2-o_2) + (|N_2|-|S_2|)\frac{w_2-o_2}{2}.
\end{align*}
and 
\begin{align*}
\SC(\opt) = |S_1| \frac{w_1-o_1}{2} + (|N_2|-|S_2|) \frac{w_2-o_2}{2}. 
\end{align*}
Hence, 
\begin{align*}
\frac{\SC(\bw)}{\SC(\opt)}
&= 1 + 2\cdot \frac{(|N_1|-|S_1|)(w_1-o_1) + |S_2|(w_2-o_2)}{|S_1| (w_1-o_1) + (|N_2|-|S_2|) (w_2-o_2)}\\
&\leq 1 + 2\cdot \frac{|N_1|-|S_1|+ |S_2|}{|S_1| + |N_2|-|S_2|}\\
&\leq 3,
\end{align*}
where the last inequality holds as $2|S_1|-|N_1|\geq 2|S_2|-|N_2|$ and the first inequality follows since $|S_1|\geq |N_1|/2$, $|S_2|\geq |N_2|/2$, and $w_2 - o_2 \leq o_2 - o_1 = w_1 - o_1$; the last is true as $s(m_2)=w_2$ and thus $m_2$, who is located at $o_2=w_1$ in this worst-case instance, is closer to $w_2$ than to $o_1$. 
\end{proof}


\subsection{General instances} \label{sec:sc-general}
To tackle the general case, we consider the following mechanism.
Let $j^* = \arg\max_{j \in [2]} |N_j\setminus N_{3-j}|$.
\begin{itemize}
    \item If $|N_1 \cap N_2| \geq |N_{j^*}\setminus N_{3-j^*}|$, then run the {\sc Median} mechanism with input the agents of $N_1 \cap N_2$ (ignoring all other agents).
    \item Otherwise, choose $w_{j^*}$ to be the candidate location closest to the median $m_{j^*}$ of $N_{j^*}\setminus N_{3-j^*}$ (we slightly abuse notation here as $m_{j^*}$ would normally be the median of $N_{j^*}$), and $w_{3-j^*}$ to be the available candidate location closest to the median $m_{3-j^*}$ of $N_{3-j^*}$; we refer to this mechanism as {\sc Alternate-Median}.
\end{itemize}
Note that {\sc Median} was shown to be strategyproof in Section \ref{sec:sc-doubleton}. As for {\sc Alternate-Median}, it is strategyproof since agents in $N_{j^*}\setminus N_{3-j^*}$ have no incentive to misreport and affect the choice of $m_{j^*}$ and $w_{j^*}$, while agents in $N_{3-j^*}$ cannot affect the choice of $w_{j^*}$ and have no incentive to misreport and affect the choice of $m_{3-j^*}$ and $w_{3-j^*}$; any misreport can only push the median, and the corresponding nearest location, farther away. Since the two cases are independent (the cardinalities of the sets of agents with different approval preferences are known), the mechanism combining {\sc Median} and {\sc Alternate-Median} is strategyproof. 

We will bound the approximation ratio of the mechanism with the following two theorems which bound the approximation ratio of the mechanism in the two cases. Without loss of generality, to simplify our notation, let $j^*=1$. 

\begin{theorem} \label{thm:sc-general-case1}
For general instances with $|N_1 \cap N_2| \geq |N_1\setminus N_2|$, 
the approximation ratio of {\sc Median} is at most $7$. 
\end{theorem}

\begin{proof}
By Theorem~\ref{thm:sc-median-3}, we have that
\begin{align*}
\sum_{i \in N_1 \cap N_2} \sum_{j \in [2]} d(i,w_j) \leq 3 \cdot  \sum_{i \in N_1 \cap N_2} \sum_{j \in [2]} d(i,o_j). 
\end{align*}
For the agents in $N_1 \setminus N_2$, by the triangle inequality and since $|N_1 \setminus N_2| \leq |N_1 \cap N_2|$, we have
\begin{align*}
\sum_{i \in N_1 \setminus N_2} d(i,w_1) 
&\leq  \sum_{i \in N_1 \setminus N_2} d(i,o_1) + |N_1 \setminus N_2| \cdot d(w_1,o_1) \\
&\leq \sum_{i \in N_1 \setminus N_2} d(i,o_1) + \sum_{i \in N_1 \cap N_2} d(w_1,o_1) \\
&\leq \sum_{i \in N_1 \setminus N_2} d(i,o_1) +  \sum_{i \in N_1 \cap N_2} \bigg(d(i,w_1) + d(i,o_1)\bigg). 
\end{align*}
Similarly, for the agents in $N_2 \setminus N_1$, since $|N_2 \setminus N_1| \leq |N_1 \setminus N_2| \leq |N_1 \cap N_2|$, we have
\begin{align*}
\sum_{i \in N_2 \setminus N_1} d(i,w_2) 
&\leq \sum_{i \in N_2 \setminus N_1} d(i,o_2) +  \sum_{i \in N_1 \cap N_2} \bigg(d(i,w_2) + d(i,o_2)\bigg). 
\end{align*}
By combining these, we have
\begin{align*}
\SC(\bw) 
&\leq 3 \cdot \SC(\opt) +  \sum_{i \in N_1 \cap N_2} \sum_{j \in [2]} \bigg( d(i,w_j) +  d(i,o_j) \bigg) \\
&\leq 3 \cdot \SC(\opt) +  4\cdot \sum_{i \in N_1 \cap N_2} \sum_{j \in [2]} d(i,o_j) \\
&\leq 7 \cdot \SC(\opt). 
\end{align*}
Therefore, the approximation ratio is at most $7$.
\end{proof}

\begin{theorem} \label{thm:sc-general-case2}
For general instances with $|N_1 \cap N_2| \leq |N_1\setminus N_2|$, the approximation ratio of {\sc Alternate-Median} is at most $7$. 
\end{theorem}

\begin{proof}
We consider the following cases:

\medskip
\noindent 
{\bf Case 1: $t(m_1) \neq t(m_2)$.}
Then, we have that $w_1 = t(m_1)$ and $w_2 = t(m_2)$. 
By the properties of the median, we have that 
$$\sum_{i \in N_1\setminus N_2} d(i,m_1) \leq \sum_{i \in N_1 \setminus N_2} d(i,x)$$
and 
$$\sum_{i \in N_2} d(i,m_2) \leq \sum_{i \in N_2} d(i,x)$$ 
for any point $x$ of the line, including $o_1$ and $o_2$. 
Also, by the definition of $t(m_j)$ for $j \in [2]$, we have that
$d(m_j,w_j) \leq d(m_j,x)$ for any candidate location $x$, including $o_j$. 
Therefore, using these facts and the triangle inequality, we bound the contribution of the different types of agents to the social cost of $\bw$. In particular, for the agents of $N_1 \setminus N_2$, we have
\begin{align*}
\sum_{i \in N_1 \setminus N_2} d(i,w_1) 
&\leq \sum_{i \in N_1 \setminus N_2} \bigg( d(i,m_1) + d(m_1,w_1) \bigg) \\
&\leq \sum_{i \in N_1 \setminus N_2} \bigg( d(i,m_1) + d(m_1,o_1) \bigg) \\
&\leq \sum_{i \in N_1 \setminus N_2} \bigg( 2\cdot d(i,m_1) + d(i,o_1) \bigg) \\
&\leq 3\cdot \sum_{i \in N_1 \setminus N_2} d(i,o_1).
\end{align*}
Similarly, for the agents of $N_2$, we have
\begin{align*}
\sum_{i \in N_2} d(i,w_2) 
\leq 3\cdot \sum_{i \in N_2} d(i,o_2).
\end{align*}
For the agents of $N_1 \cap N_2$ in terms of $w_1$, using the triangle inequality, we obtain
\begin{align*}
\sum_{i \in N_1 \cap N_2} d(i,w_1) 
&\leq \sum_{i \in N_1 \cap N_2} d(i,o_1) + |N_1 \cap N_2|\cdot d(w_1,o_1) \\
&\leq \sum_{i \in N_1 \cap N_2} d(i,o_1) + \sum_{i \in N_1 \setminus N_2} d(w_1,o_1) \\
&= \sum_{i \in N_1 \cap N_2} d(i,o_1) + \sum_{i \in N_1 \setminus N_2} \bigg( d(i,w_1) + d(i,o_1) \bigg) \\
&\leq \sum_{i \in N_1 \cap N_2} d(i,o_1) + 4\cdot \sum_{i \in N_1 \setminus N_2} d(i,o_1).
\end{align*}
By putting everything together, we obtain an upper bound of $7$.

\medskip
\noindent
{\bf Case 2: $t(m_1) = t(m_2)$.} 
In this case, we have that $w_1 = t(m_1) = t(m_2)$ and $w_2 = s(m_2)$. Clearly, if $d(m_2,w_2) \leq d(m_2,o_2)$, we get an upper bound of $7$, similarly to Case 1. So, we can assume that $d(m_2,w_2) > d(m_2,o_2)$, which combined with the fact that $w_2 = s(m_2)$, implies that $o_2 = t(m_2) = w_1$. 
For the agents in $N_1 \setminus N_2$, since $w_1 = t(m_1)$, we have a $3$-approximation guarantee (using the same arguments as above):
\begin{align*}
\sum_{i \in N_1 \setminus N_2} d(i,w_1) \leq 3 \cdot \sum_{i \in N_1 \setminus N_2} d(i,o_1).
\end{align*}
For the agents in $N_1 \cap N_2$ in terms of $w_1$, similarly to Case 1, we have
\begin{align*}
\sum_{i \in N_1 \cap N_2} d(i,w_1) \leq \sum_{i \in N_1 \cap N_2} \bigg( d(i,o_1) + d(o_1,w_1) \bigg) = \sum_{i \in N_1 \cap N_2} d(i,o_1)  + |N_1 \cap N_2| \cdot d(o_1,w_1).
\end{align*}
For the agents in $N_2 = (N_2\setminus N_1) \cup (N_1 \cap N_2)$ in terms of $w_2$, since $d(m_2,w_2) \leq d(m_2,o_1)$, $w_1 = o_2$, and $m_2$ minimizes the total distance of the agents in $N_2$ from any other point of the line, by the triangle inequality, we have
\begin{align*}
\sum_{i \in N_2} d(i,w_2) 
&\leq \sum_{i \in N_2} \bigg( d(i,m_2) + d(m_2,w_2) \bigg) \\
&\leq \sum_{i \in N_2} \bigg( d(i,m_2) + d(m_2,o_1) \bigg) \\
&\leq \sum_{i \in N_2} \bigg( d(i,m_2) + d(m_2,o_2) + d(o_1,o_2) \bigg) \\
&\leq \sum_{i \in N_2} \bigg( 2d(i,m_2) + d(i,o_2) + d(o_1,w_1) \bigg) \\
&\leq 3 \cdot \sum_{i \in N_2} d(i,o_2)  + |N_2| \cdot d(o_1,w_1).
\end{align*}
So, by putting everything together and using the fact that $|N_2| = |N_2 \setminus N_1| + |N_2 \cap N_1|$, we have
\begin{align*}
\SC(\bw) 
&\leq 3 \cdot \sum_{i \in N_1 \setminus N_2} d(i,o_1) + 3 \cdot \sum_{i \in N_2} d(i,o_2) + \sum_{i \in N_1 \cap N_2} d(i,o_1) \\
&\quad + |N_1 \cap N_2| \cdot d(o_1,w_1) + |N_2| \cdot d(o_1,w_1) \\
&\leq 3 \cdot \SC(\opt) + \bigg( |N_2 \setminus N_1| + 2\cdot |N_1 \cap N_2| \bigg) \cdot d(o_1,w_1).
\end{align*}
Since $w_1 = t(m_1)$, half of the agents in $N_1 \setminus N_2$ suffer a cost of at least $d(o_1,w_1)/2$ in the optimal solution. Also, all the agents of $N_1 \cap N_2$ suffer a cost of at least $d(o_1,o_2)/2 = d(o_1,w_1)/2$, and thus 
\begin{align*}
    \SC(\opt) \geq \bigg( \frac{|N_1 \setminus N_2|}{4} + \frac{|N_1 \cap N_2|}{2} \bigg) d(o_1,w_1).
\end{align*}
Hence, since $|N_2 \setminus N_1| \leq |N_1 \setminus N_2|$, the approximation ratio is at most
\begin{align*}
 3 + 4\cdot \frac{|N_2 \setminus N_1| + 2|N_1 \cap N_2|}{|N_1 \setminus N_2| + 2|N_1 \cap N_2|} \leq 7. 
\end{align*}
Consequently, the approximation ratio is overall at most $7$. 
\end{proof}

Using Theorem~\ref{thm:sc-general-case1} and Theorem~\ref{thm:sc-general-case2}, we obtain the following result. 

\begin{corollary}
For general instances, there is a strategyproof mechanism with approximation ratio at most $7$. 
\end{corollary}


\section{Max cost} \label{sec:max}
In this section, we turn our attention to the max cost objective for which we show that the best possible approximation ratio of strategyproof mechanisms is between $2$ and $3$ for doubleton instances, and exactly $3$ for singleton and general preferences. 

\subsection{Doubleton instances} \label{sec:max-doubleton}
For the upper bound, we consider a simple mechanism that places both facilities at the candidate locations that are closest to the leftmost agent $\ell$. We refer to this mechanism as {\sc Leftmost}; see Mechanism~\ref{mech:max-leftmost-rightmost}. It is not hard to show that this mechanism is strategyproof and that it achieves an approximation ratio of $3$. 

\SetCommentSty{mycommfont}
\begin{algorithm}[h]
\SetNoFillComment
\caption{{\sc Leftmost}}
\label{mech:max-leftmost-rightmost}
{\bf Input:} Reported positions of agents\;
{\bf Output:} Facility locations $\bw = (w_1,w_2)$ \;
$\ell \gets$ leftmost agent in $N_1 \cap N_2$\;
$w_1 \gets t(\ell)$\;
$w_2 \gets s(\ell)$\;
\end{algorithm}

\begin{theorem} \label{thm:max-all-both-upper-2}
For doubleton instances, {\sc Leftmost} is strategyproof and achieves an approximation ratio of at most $3$.
\end{theorem}

\begin{proof}
For the strategyproofness of the mechanism, consider any agent $i$; recall that $i$ approves both facilities. To affect the outcome, agent $i$ would have report a position that lies at the left of $\ell$. However, changing the leftmost agent position can only lead to placing the facilities at locations farther away from $i$, and hence $i$ has no incentive to misreport. 

For the approximation ratio, let $\opt = (o_1,o_2)$ be an optimal solution. Clearly, there exist $x \in \{o_1,o_2\}$ and $y \in \{o_1,o_2\}\setminus \{x\}$ such that $d(\ell,w_1) \leq d(\ell,x)$ and $d(r,w_2) \leq d(\ell,y)$. Let $i$ be the (rightmost) agent who determines the max cost of the mechanism. Using the triangle inequality, we have
\begin{align*}
\MC(\bw)=d(i,w_1) + d(i,w_2) 
&\leq 
\bigg(d(i,x) + d(\ell,x) + d(\ell,w_1) \bigg) + \bigg( d(i,y) + d(\ell,y) + d(\ell,w_2) \bigg) \\
&\leq 3 \max_{j \in N} \bigg(d(j,x) + d(j,y)\bigg) = 3 \cdot \MC(\opt).
\end{align*}
Therefore, the approximation ratio is at most $3$. 
\end{proof}

We next show a slightly weaker lower bound of $2$ on the approximation ratio of any strategyproof mechanism. 

\begin{theorem} \label{max:both-lower-2}
For doubleton instances, the approximation ratio of any deterministic strategyproof mechanism is at least $2-\delta$, for any $\delta > 0$. 
\end{theorem}

\begin{proof}
Consider the following instance $I_1$: 
There are three candidate locations at $-1$, $0$, and $1$ and two agents (that approve both facilities) positioned at $-\varepsilon$ and $\varepsilon$, respectively, for some infinitesimal $\varepsilon > 0$. Since there are two facilities to be located, at least one of them must be placed at $-1$ or $1$; see Figure~\ref{fig:max-doubleton--lower-a}. Without loss of generality, let us assume that a facility is placed at $1$. 

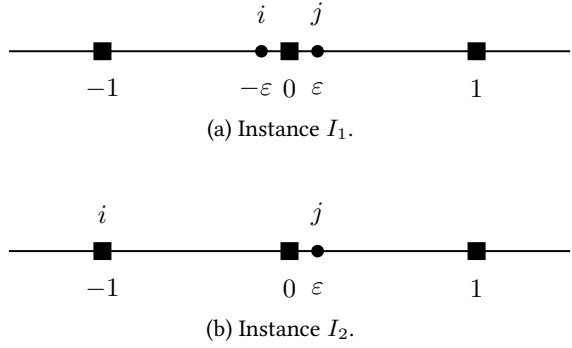
\begin{figure}[t]
\tikzset{every picture/.style={line width=0.75pt}} 
\centering
\begin{subfigure}[t]{0.45\linewidth}
\centering
\begin{tikzpicture}[x=0.7pt,y=0.7pt,yscale=-1,xscale=1]
\draw [line width=0.75]  (0,0) -- (300,0) ;
\filldraw ([xshift=-3pt,yshift=-3pt]50,0) rectangle ++(6pt,6pt);
\filldraw ([xshift=-3pt,yshift=-3pt]250,0) rectangle ++(6pt,6pt);
\filldraw ([xshift=-3pt,yshift=-3pt]150,0) rectangle ++(6pt,6pt);
\filldraw (165,0) circle (2pt);
\filldraw (135,0) circle (2pt);

\draw (135,-20) node [inner sep=0.75pt]  [font=\small]  {$i$};
\draw (165,-20) node [inner sep=0.75pt]  [font=\small]  {$j$};

\draw (50,20) node [inner sep=0.75pt]  [font=\small]  {$-1$};
\draw (250,20) node [inner sep=0.75pt]  [font=\small]  {$1$};
\draw (150,20) node [inner sep=0.75pt]  [font=\small]  {$0$};
\draw (165,20) node [inner sep=0.75pt]  [font=\small]  {$\varepsilon$};
\draw (132,20) node [inner sep=0.75pt]  [font=\small]  {$-\varepsilon$};
\end{tikzpicture}
\caption{Instance $I_1$.}
\label{fig:max-doubleton--lower-a}
\end{subfigure}
\\[20pt]
\begin{subfigure}[t]{0.45\linewidth}
\centering
\begin{tikzpicture}[x=0.7pt,y=0.7pt,yscale=-1,xscale=1]
\draw [line width=0.75]  (0,0) -- (300,0) ;
\filldraw ([xshift=-3pt,yshift=-3pt]50,0) rectangle ++(6pt,6pt);
\filldraw ([xshift=-3pt,yshift=-3pt]250,0) rectangle ++(6pt,6pt);
\filldraw ([xshift=-3pt,yshift=-3pt]150,0) rectangle ++(6pt,6pt);
\filldraw (165,0) circle (2pt);

\draw (50,-20) node [inner sep=0.75pt]  [font=\small]  {$i$};
\draw (165,-20) node [inner sep=0.75pt]  [font=\small]  {$j$};

\draw (50,20) node [inner sep=0.75pt]  [font=\small]  {$-1$};
\draw (250,20) node [inner sep=0.75pt]  [font=\small]  {$1$};
\draw (150,20) node [inner sep=0.75pt]  [font=\small]  {$0$};
\draw (165,20) node [inner sep=0.75pt]  [font=\small]  {$\varepsilon$};
\end{tikzpicture}
\caption{Instance $I_2$.}
\label{fig:max-doubleton--lower-b}
\end{subfigure}
\caption{The two instances used in the proof of the lower bound of $2$ in terms of the max cost for doubleton instances (Theorem~\ref{max:both-lower-2}). Both agents $i$ and $j$ approve both facilities. Rectangles represent candidate locations.}
\label{fig:max-doubleton--lower}
\end{figure}

Now, consider the instance $I_2$, which is the same as $I_1$ with the only difference that the agent at $-\varepsilon$ has been moved to $-1$; see Figure~\ref{fig:max-doubleton--lower-b}. To maintain strategyproofness, a facility must be placed at $1$ in $I_2$ as well; otherwise, the agent at $-\varepsilon$ in $I_1$ would misreport her location as $-1$ to affect the outcome and decrease her cost. So, in $I_2$, any strategyproof mechanism either places one facility at $-1$ and one facility at $1$, for a max cost of $2$, or one facility at $0$ and one facility at $1$, for a max cost of $3$. However, placing one facility at $-1$ and one facility at $0$ leads to max cost $1+\varepsilon$, and thus an approximation ratio of at least $2-\delta$, for any $\delta > 0$.
\end{proof}


\subsection{Singleton instances} \label{sec:max-singleton}
As argued at the beginning of Section~\ref{sec:sc-singleton}, instances in which all agents approve one of the facilities are equivalent to having just this one facility to place. Consequently, by the work of \citet{Tang2020candidate}, we cannot hope to achieve an approximation ratio better than $3$ for singleton instances. For completeness, we include a simple proof of this lower bound here. 

\begin{theorem} \label{max:singleton-lower-3}
For singleton instances, the approximation ratio of any strategyproof mechanism is at least $3-\delta$, for any $\delta > 0$.     
\end{theorem}

\begin{proof}
Consider the following instance $I_1$: 
There are two candidate locations at $-1$ and $1$ and two agents approving only $F_1$ positioned at $-\varepsilon$ and $\varepsilon$, respectively, for some infinitesimal $\varepsilon > 0$; see Figure~\ref{fig:max-singleton-lower-a}. Without loss of generality, we can assume that $F_1$ is placed at $1$ and $F_2$ at $-1$. 

\begin{figure}[t]
\tikzset{every picture/.style={line width=0.75pt}} 
\centering
\begin{subfigure}[t]{\linewidth}
\centering
\begin{tikzpicture}[x=0.7pt,y=0.7pt,yscale=-1,xscale=1]
\draw [line width=0.75]  (0,0) -- (400,0);
\filldraw ([xshift=-3pt,yshift=-3pt]200,0) rectangle ++(6pt,6pt);
\filldraw ([xshift=-3pt,yshift=-3pt]350,0) rectangle ++(6pt,6pt);
\filldraw (260,0) circle (2pt);
\filldraw (290,0) circle (2pt);

\draw (260,-20) node [inner sep=0.75pt]  [font=\small]  {$i$};
\draw (290,-20) node [inner sep=0.75pt]  [font=\small]  {$j$};

\draw (200,20) node [inner sep=0.75pt]  [font=\small]  {$-1$};
\draw (350,20) node [inner sep=0.75pt]  [font=\small]  {$1$};
\draw (290,20) node [inner sep=0.75pt]  [font=\small]  {$\varepsilon$};
\draw (260,20) node [inner sep=0.75pt]  [font=\small]  {$-\varepsilon$};
\end{tikzpicture}
\caption{Instance $I_1$.}
\label{fig:max-singleton-lower-a}
\end{subfigure}
\\[20pt]
\begin{subfigure}[t]{\linewidth}
\centering
\begin{tikzpicture}[x=0.7pt,y=0.7pt,yscale=-1,xscale=1]
\draw [line width=0.75]  (0,0) -- (400,0) ;
\filldraw ([xshift=-3pt,yshift=-3pt]200,0) rectangle ++(6pt,6pt);
\filldraw ([xshift=-3pt,yshift=-3pt]350,0) rectangle ++(6pt,6pt);
\filldraw (50,0) circle (2pt);
\filldraw (290,0) circle (2pt);

\draw (50,-20) node [inner sep=0.75pt]  [font=\small]  {$i$};
\draw (290,-20) node [inner sep=0.75pt]  [font=\small]  {$j$};

\draw (50,20) node [inner sep=0.75pt]  [font=\small]  {$-2$};
\draw (200,20) node [inner sep=0.75pt]  [font=\small]  {$-1$};
\draw (350,20) node [inner sep=0.75pt]  [font=\small]  {$1$};
\draw (290,20) node [inner sep=0.75pt]  [font=\small]  {$\varepsilon$};
\end{tikzpicture}
\caption{Instance $I_2$.}
\label{fig:max-singleton-lower-b}
\end{subfigure}
\caption{The two instances used in the proof of the lower bound of $3$ in terms of the max cost for singleton instances (Theorem~\ref{max:singleton-lower-3}). 
Both agents $i$ and $j$ approve facility $F_1$. Rectangles represent candidate locations.}
\label{fig:max-singleton-lower}
\end{figure}
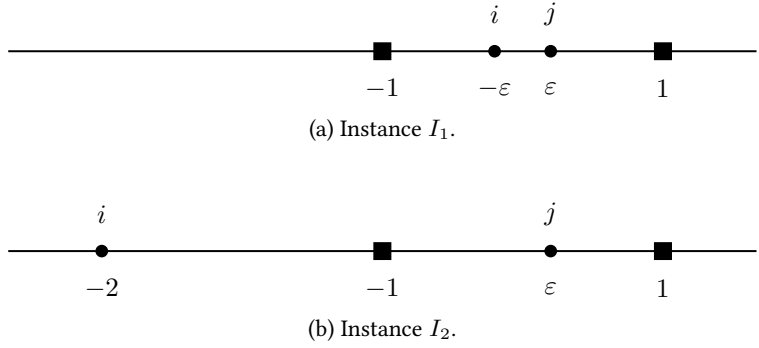

Now, consider the instance $I_2$, which is the same as $I_1$ with the only difference that the agent at $-\varepsilon$ has been moved to $-2$; see Figure~\ref{fig:max-singleton-lower-b}. To maintain strategyproofness, $F_1$ must be placed at $1$ in $I_2$ as well; otherwise, the agent at $-\varepsilon$ in $I_1$ would misreport her position as $-2$ to decrease her cost from $1+\varepsilon$ to $1-\varepsilon$. This leads to a max cost of $3$, when a max cost of $1+\varepsilon$ is possible by placing $F_1$ at $-1$. Therefore, the approximation ratio is at least $3-\delta$, for any $\delta > 0$.
\end{proof}

A first idea towards an upper bound could be to place $F_1$ at the closest candidate location to the leftmost agent $\ell_1$ of $N_1$, and $F_2$ at the closest available candidate location to an agent of $N_2$, such as the leftmost agent $\ell_2$ or the rightmost agent $r_2$. While these  mechanisms are clearly strategyproof, it is not hard to observe that they cannot achieve a good enough approximation ratio. 

\begin{example}
If we place $F_1$ at $t(\ell_1)$ and $F_2$ at $t(r_2)$ or $s(r_2)$ depending on availability, then consider the following instance: 
There are three candidate locations at $0$, $2$, and $6$. 
There is an agent $\ell_1$ that approves $F_1$ at $1+\varepsilon$, an agent $\ell_2$ that approves $F_2$ at $1$, and another agent $r_2$ that approves $F_2$ at $3+\varepsilon$, for some infinitesimal $\varepsilon>0$. So, we place $F_1$ at $2$ and $F_2$ at $6$ for a max cost of $5$ (determined by $\ell_2$). On the other hand, we could place $F_1$ at $0$ and $F_2$ at $2$ for a max cost of approximately $1$, leading to an approximation ratio of $5$. Clearly, if we chose $\ell_2$ instead of $r_2$ to determine the location of $F_2$, there is a symmetric instance leading again to the same lower bound. 
\end{example}

The above example illustrates that it is not always a good idea to choose a priori $\ell_2$ or $r_2$ to determine where to place $F_2$, especially when the closest candidate location to them might not be available after placing $F_1$. Instead, we need to carefully decide whether $\ell_2$ or $r_2$ or neither of them is the best one to choose where to place $F_2$. We make this decision as follows: We ``ask'' $\ell_2$ and $r_2$ to ``vote'' over two candidate locations; the candidate location $L$ that is the closest at the left of $t(\ell_1)$ (where $F_1$ is placed) and the candidate location $R$ that is the closest at the right of $t(\ell_1)$. If $\ell_2$ and $r_2$ ``agree'', then they are both on the same side of the midpoint of the interval defined by $L$ and $R$, and thus depending on whether they are on the left side (agree on $L$) or the right side (agree on $R$), we allow $r_2$ or $\ell_2$, respectively, to make the choice of where to place $F_2$. If they ``disagree'', they are on different sides of the interval's midpoint, so neither $\ell_2$ nor $r_2$ should make a choice of where to place $F_2$; in this case, the closest of $L$ and $R$ to $t(\ell_1)$ is a good candidate location to place $F_2$. 
This idea is formalized in Mechanism~\ref{mech:maxV}, which we call $\maxV$. 

\SetCommentSty{mycommfont}
\begin{algorithm}[h]
\SetNoFillComment
\caption{$\maxV$}
\label{mech:maxV}
{\bf Input:} Reported positions of agents with singleton preferences\;
{\bf Output:} Facility locations $\bw = (w_1,w_2)$\;
$\ell_1 \gets$ leftmost agent in $N_1$\;
$\ell_2 \gets$ leftmost agent in $N_2$\;
$r_2 \gets$ rightmost agent in $N_2$\; 
$w_1 \gets t(\ell_1)$\;
$L \gets$ closest candidate location at the left of $w_1$\;
$R \gets$ closest candidate location at the right of $w_1$\;
\tcp*[h]{{\bf (case 1)} $\ell_2$ and $r_2$ agree that $L$ is closer, so $r_2$ gets to choose} \\
\uIf{$\ell_2$ and $r_2$ are both closer to $L$ than to $R$}{ 
    \uIf{$t(r_2)$ is available}{ 
       $w_2 \gets t(r_2)$\;
    }
    \Else{
       $w_2 \gets s(r_2)$\;
    }
}
\tcp*[h]{{\bf (case 2)} $\ell_2$ and $r_2$ agree that $R$ is closer, so $\ell_2$ gets to choose} \\
\uElseIf{$\ell_2$ and $r_2$ are both closer to $R$ than to $L$}{ 
    \uIf{$t(\ell_2)$ is available}{ 
       $w_2 \gets t(\ell_2)$\;
    }
    \Else{
       $w_2 \gets s(\ell_2)$\;
    }
}
\tcp*[h]{{\bf (case 3)} $\ell_2$ and $r_2$ disagree, so choose the closest of $L$ and $R$ to $w_1$} \\
\Else{
    $w_2 \gets \arg\min_{x \in \{L,R\}}\{|w_1-x|\}$\;
}
\end{algorithm}

We first show that $\maxV$ is strategyproof. 

\begin{theorem}
For singleton instances, $\maxV$ is strategyproof. 
\end{theorem}

\begin{proof}
Observe that no agent in $N_1$ has incentive to misreport as facility $F_1$ is located at the closest candidate location to $\ell_1$; indeed, $\ell_1$ is content while no agent would like to misreport to become the leftmost agent of $N_1$ as then $F_1$ will either remain at the same location or could be moved farther away. For the agents of $N_2$, we consider each case separately depending on which is the true profile.
Denote by $w_2$ the location of $F_2$ when the agents report their positions truthfully.

\medskip
\noindent 
{\bf (Case 1)}
Clearly, $r_2$ has no incentive to deviate. Consider an agent $i \in N_2$, other than $r_2$, that deviates and misreports a position $z$. 
\begin{itemize}
\item
If $z \leq r_2$, then the location $F_2$ is still $w_2$. 

\item 
If $z  \in \left(r_2, \frac{L+R}{2}\right]$, agent $i$ becomes the rightmost agent but we are still in Case 1. So, the location of $F_2$ becomes the closest available location to $z$, which is either $w_2$ or some candidate location at the right of $w_2$. This means that the cost of $i$ either remains the same or increases, and thus $i$ has no incentive to misreport such a position.

\item 
If $z > \frac{L+R}{2}$, agent $i$ becomes the rightmost agent and the location of $F_2$ is determined by Case 3, i.e., becomes $y = \arg\min_{x \in \{L,R\}}\{|w_1-x|\}$. Since the true rightmost agent $r_2$ is closer to $L$ than to $R$, it holds that $w_2 \leq L \leq y$. This again means that the cost of $i$ either remains the same or increases, and thus $i$ has no incentive to misreport such a position.
\end{itemize}

\medskip
\noindent 
{\bf (Case 2)}
This is symmetric to Case 1. 

\medskip
\noindent 
{\bf (Case 3)} 
Observe that any deviation that still leads to Case 3 does not affect the outcome of the mechanism as $w_2 =  \arg\min_{x \in \{L,R\}}\{|w_1-x|\}$. Hence, no agent $i \in N_2 \setminus\{\ell_2,r_2\}$ can affect the outcome as any possible misreported position can either be at the left of $\ell_2$ or the right of $r_2$, which means that we are still in Case 3. Now, let us assume that $r_2$ misreports so that the location of $F_2$ is determined by Case 1. Since, in that case, all agents are closer to $L$ than to $R$, and there are no other available candidate locations in the interval $[L,R]$ (since $w_1$ is occupied by $F_1$), $F_2$ can only be placed at some location $y \leq L$, which is clearly not better for $r_2$. A symmetric argument for $\ell_2$ shows that again no agent can misreport. 
\end{proof}

Next, we show that $\maxV$ achieves an approximation ratio of at most $3$.

\begin{theorem} \label{thm:max-singleton-upper}
For singleton instances, the approximation ratio of $\maxV$ is at most $3$. 
\end{theorem}

\begin{proof}
If the max cost of the mechanism is due to an agent $i \in N_1$, the choice $w_1 = t(\ell_1)$ implies that $d(\ell_1,w_1)\leq d(\ell_1,o_1)$, and thus, by the triangle inequality, we have that
\begin{align*}
\MC(\bw) = d(i,w_1) 
    \leq d(i,o_1) + d(\ell_1,o_1) + d(\ell_1,w_1)
     \leq d(i,o_1) + 2\cdot d(\ell_1,o_1)
    \leq 3 \cdot \MC(\opt). 
\end{align*}
So, we now focus on the case where the max cost of the mechanism is determined by an agent in $N_2$ and we may assume that $w_2\neq o_2$ as otherwise the claim holds trivially. Due to the symmetry of Case 1 and Case 2, it suffices to bound the approximation ratio in Case 1 and in Case 3. In any of these cases, if there is an agent of $N_2$ that is closer to $w_2$ than to $o_2$, then, similarly to above, by applying the triangle inequality, we can again show that the approximation ratio is at most $3$. Thus, we will assume that all agents of $N_2$ are closer to $o_2$ than to $w_2$, which means that $w_1=o_2$. To see that, note that, in Case 1, $o_2$ has to be  unavailable  as it must hold $t(r_2)=o_2$, while in Case 3,  $w_2$ is the closest candidate location among $L$ and $R$ to $w_1$ (and thus $\ell_2$ cannot be at the left of $L$ if $L$ is chosen and $r_2$ cannot be at the right of $R$ if $R$ is chosen). Due to this, $o_1$ cannot be $w_1$ and we have the following two possibilities:
\begin{itemize}
    \item If $o_1 \leq L$, then $d(\ell_1,o_1) \geq d(\ell_1,L)$.
    \item If $o_1 \geq R$, then $d(\ell_1,o_1) \geq d(\ell_1,R)$.
\end{itemize}
\medskip
\noindent 
{\bf (Case 1)}
Since $t(r_2)=w_1=o_2$ and $w_2$ is the closest available candidate location to $r_2$, it has to be the case that $w_2=L$. Let $i \in \{\ell_2,r_2\}$ be the agent of $N_2$ that gives the max cost. 
\begin{itemize}
\item If $o_1 \leq L$, then due to the triangle inequality, and the facts that $o_2=w_1$ and $d(\ell_1,o_1) \geq d(\ell_1,L)$, we have
\begin{align*}
    \MC(\bw) 
    = d(i,L) 
    &\leq d(i,o_2) + d(\ell_1,o_2) + d(\ell_1,L) \\
    &= d(i,o_2) + d(\ell_1,w_1) + d(\ell_1,L) \\
    &\leq d(i,o_2) + d(\ell_1,o_1) + d(\ell_1,o_1) \\ 
    &\leq 3 \cdot \MC(\opt).
\end{align*}

\item If $o_1 \geq R$, then due to the triangle inequality, and the facts that $o_2=w_1$, $d(i,L) \leq d(i,R)$ and  $d(\ell_1,o_1) \geq d(\ell_1,R)$, we have
\begin{align*}
    \MC(\bw) 
    = d(i,L) 
    \leq d(i,R) 
    &\leq d(i,o_2) + d(\ell_1,o_2) + d(\ell_1,R) \\
    &= d(i,o_2) + d(\ell_1,w_1) + d(\ell_1,R) \\
    &\leq d(i,o_2) + d(\ell_1,o_1) + d(\ell_1,o_1) \\ 
    &\leq 3 \cdot \MC(\opt).
\end{align*}
\end{itemize}

\medskip
\noindent 
{\bf (Case 3)}
Without loss of generality, let us assume that $w_2 = R$; the case where $w_2 = L$ is symmetric. 
So, $d(R,w_1) = d(R,o_2) \leq d(L,o_2) = d(L,w_1)$. Since $r_2 \leq R$, the max cost of the mechanism is determined by agent $\ell_2$.
\begin{itemize}
    \item If $o_1 \geq R$, then due to the triangle inequality, and the facts that $w_1=o_2$ and $d(\ell_1,o_1) \geq d(\ell_1,R)$, we have
    \begin{align*}
        \MC(\bw) = d(\ell_2, R) 
        &\leq d(\ell_2,o_2) + d(\ell_1,o_2) + d(\ell_1,R) \\
        &= d(\ell_2,o_2) + d(\ell_1,w_1) + d(\ell_1,R) \\
        &\leq d(\ell_2,o_2) + d(\ell_1,o_1) + d(\ell_1,o_1) \\ 
        &\leq 3 \cdot \MC(\opt).
    \end{align*}

    \item If $o_1 \leq L$, then due to the triangle inequality, and the facts that $d(R,o_2) \leq d(L,o_2)$, $w_1=o_2$ and $d(\ell_1,o_1) \geq d(\ell_1,L)$, we have
    \begin{align*}
        \MC(\bw) = d(\ell_2, R) 
        &\leq d(\ell_2,o_2) + d(R,o_2) \\
        &\leq d(\ell_2,o_2) + d(L,o_2) \\
        &\leq d(\ell_2,o_2) + d(\ell_1,o_2) + d(\ell_1,L) \\
        &= d(\ell_2,o_2) + d(\ell_1,w_1) + d(\ell_1,L) \\
        &\leq d(\ell_2,o_2) + d(\ell_1,o_1) + d(\ell_1,o_1) \\ 
        &\leq 3 \cdot \MC(\opt).
    \end{align*}
\end{itemize}
This completes the proof. 
\end{proof}


\subsection{General instances} \label{sec:max-general}
To tackle the general case, we consider a mechanism that runs {\sc Leftmost} in case the instance consists of at least one agent with doubleton preference, and $\maxV$ in case the instance is singleton. 
It is not hard to observe that {\sc Leftmost} is strategyproof even when there are agents with singleton preference; its decision is fully determined by the leftmost agent with doubleton preference and the input of any other agent is ignored. Hence, the mechanism is overall strategyproof. We will now show that {\sc Leftmost} still achieves an approximation ratio of at most $3$ when it is applied, which will allow us to show an overall bound of $3$. 

\begin{theorem} \label{thm:max-at-least-one-both-upper-3}
For instances with at least one agent with doubleton preference, the approximation ratio of {\sc Leftmost} is at most $3$.
\end{theorem}

\begin{proof}
We consider cases depending on the preference of the agent $i$ that determines the max cost of the mechanism. 
Let $\ell$ be the leftmost agent in $N_1 \cap N_2$, and recall that $w_1 = t(\ell)$ and $w_2 = s(\ell)$. 

\medskip
\noindent
{\bf (Case 1)} The max cost is determined by an agent $i \in N_1 \setminus N_2$. 
Then, by the triangle inequality and since $d(\ell,w_1) \leq d(\ell,o_2)$, we have
\begin{align*}
\MC(\bw) = d(i,w_1) 
&\leq d(i,o_1) + d(\ell,o_1) + d(\ell,w_1) \\
&\leq d(i,o_1) + d(\ell,o_1) + d(\ell,o_2) \\
&\leq 2\cdot \MC(\opt). 
\end{align*}

\medskip
\noindent
{\bf (Case 2)} The max cost is determined by an agent $i \in N_2 \setminus N_1$. 
Since $w_2=s(\ell)$, there exists $x \in \{o_1,o_2\}$ such that $d(\ell,w_2) \leq d(\ell,x) \leq \MC(\opt)$. 
Hence, by the triangle inequality, we have
\begin{align*}
\MC(\bw) = d(i,w_2) 
\leq d(i,o_2) + d(\ell,o_2) + d(\ell,w_2) 
\leq 3\cdot \MC(\opt). 
\end{align*}

\medskip
\noindent
{\bf (Case 3)} The max cost is determined by an agent $i \in N_1 \cap N_2$. Then, following the proof of Theorem~\ref{thm:max-all-both-upper-2} for doubleton instances, we can show an upper bound of $3$. 
\end{proof}

By combining Theorem~\ref{thm:max-at-least-one-both-upper-3} and Theorem~\ref{thm:max-singleton-upper}, we obtain the following result. 

\begin{corollary}
For general instances, there is a strategyproof mechanism with approximation ratio at most $3$. 
\end{corollary}


\section{Allowing same facility locations} \label{sec:same}
In this last section we explore the simpler model in which {\em the two facilities can be placed at the same candidate location}. We show tight bounds on the approximation ratio of deterministic mechanisms for doubleton and general instances (we will not consider singleton instances separately as the approximation ratio turns out to be exactly the same as for general instances). Our results for this model are summarized in Table~\ref{tab:same:results}. 

\begin{table}[t]
    \centering
    \begin{tabular}{c|cc}
                      & Social cost & Max cost \\ 
       \hline
       Doubleton   & $1+\sqrt{2}$ & $2$ \\
       General     & $3$ & $3$ \\
       \hline
    \end{tabular}
    \caption{Overview of the tight bounds for the model where the two facilities are allowed to be placed at the same candidate location.}
    \label{tab:same:results}
\end{table}

All the mechanisms we will consider in this section place the facilities at the closest locations to some fixed agents that approve them. In particular, given the positions reported by the agents, for some $q_j \in [n]$, we place each facility $F_j$ at $t(i_j)$, where $i_j$ is the $q_j$-th ordered agent in $N_j$. It is not hard to verify that all such mechanisms are strategyproof. Indeed, to change the outcome of the mechanism, an agent in $N_j$ would have to report a position that changes the $q_j$-th ordered agent in $N_j$, but this would mean that the facilities that this agent approves might move farther away from the true position of the agent. 

Before we continue we remark that the fact that facilities can be placed at the same location is crucial for our mechanisms to be strategyproof since this eliminates possible misreports by the $q_j$-th ordered agents who determine where the facilities are placed. To be more specific, suppose that we try to adapt this mechanism for the main model that we considered in the previous sections in which the two facilities can only be placed at different locations. Then, in case $t(i_1) = t(i_2)$ we would have to resolve this collision somehow, for example by giving priority to one of these agents, say $i_1$, and placing $F_1$ at $w_1 =t(i_1)$ and then $F_2$ at some other location $w_2$ that is a function of $i_2$ such as $s(i_2)$. However, if $i_1$ approves both facilities, it might be the case that $w_2$ is not close to her position, and thus she prefers to misreport that she is closer to $s(i_1)$ rather than $t(i_1)$, leading to $F_1$ being placed at $s(i_1)$ and then $F_2$ at $t(i_2)=t(i_1)$. Such misreports cannot happen when facilities are allowed to be placed at the same location. 

\subsection{Social cost}\label{sec:same:social}
We start with the case of doubleton instances for which we show a tight bound of $1+\sqrt{2}$. The lower bound follows by observing that the proof of Theorem~\ref{thm:sc-doubleton-lower} holds even when facilities are allowed to be placed at the same location; in particular, in the proof of that theorem we made the simplification that there are two candidate location at $2$ and $-2$, thus having capacity for both facilities.
For the upper bound, first observe that the {\sc Median} mechanism from Section~\ref{sec:sc-doubleton} can also be adapted to the current model (by placing both facilities to the location closest to the median agent), but it is not hard to show that it still cannot achieve an approximation ratio better than $3$. To improve upon the bound of $3$, we consider a family of mechanisms, which, for a parameter $\alpha \in (0,1/2)$, place one facility at the candidate location closest to the position reported by the $\alpha n$-leftmost agent, and the other facility at the candidate location closest to the position reported by the $(1-\alpha)n$-leftmost agent\footnote{Formally, it would be the $\lceil \alpha n \rceil$-leftmost and the $\lceil (1-\alpha)n\rceil$-leftmost agent, respectively, and we require that $\lceil \alpha n \rceil < \lceil (1-\alpha) n \rceil$. This can be guaranteed by creating an identical number of copies for each agent and running the mechanism on the modified instance; the approximation ratio for the modified instance is exactly the same as for the original instance. We omit the ceilings to make the exposition clearer.}. We refer to such mechanisms as $\alphaMech$; see Mechanism~\ref{mech:sc-alpha} for a description. It is not hard to observe that, for any $\alpha \in (0,1/2)$, the mechanism is strategyproof since it falls within the class of mechanisms we described earlier with $q_1 = \alpha n$ and $q_2 = (1-\alpha) n$. 
We now focus on bounding the approximation ratio.

\SetCommentSty{mycommfont}
\begin{algorithm}[h]
\SetNoFillComment
\caption{$\alphaMech$}
\label{mech:sc-alpha}
{\bf Input:} Reported positions of agents with doubleton preferences\;
{\bf Output:} Facility locations $\bw = (w_1,w_2)$\;
$i \gets \alpha n$-leftmost agent\;
$j \gets (1-\alpha)n$-leftmost agent\;
$w_1 \gets t(i)$\;
$w_2 \gets t(j)$\;
\end{algorithm}

\begin{theorem}
For doubleton instances, the approximation ratio of \text{\sc $\left(\sqrt{2}-1\right)$-Statistic} is at most $1+\sqrt{2}$.
\end{theorem}

\begin{proof}
We have $\alpha=\sqrt{2}-1$ and note that $\frac{1+\alpha}{1-\alpha} = \frac{1}{\alpha} = 1+\sqrt{2}$. 
If $o$ is a location that minimizes the total distance from the agent positions, then for any $x$ and $y$ such that $o \leq x \leq y$ or $y \leq x \leq o$, it holds that $\sum_{i\in N} d(i,o) \leq \sum_{i \in N} d(i,x) \leq \sum_{i\in N} d(i,y)$. Hence, since the individual cost of each agent is the sum of distances from both facilities, there exists an optimal solution $\opt=(o_1,o_2)$ such that $o_1=o_2=o$. Without loss of generality, we assume that $w_1 \leq w_2$, and it must be the case that $w_1 \neq o$ or $w_2 \neq o$ since otherwise the approximation ratio would be $1$. 
We consider the following cases:

\medskip
\noindent
{\bf Case 1: $w_1 < o = w_2$ (the case $w_1 = o < w_2$ is symmetric).} \\
By the definition of the mechanism, there is a set $S$ of $\alpha n$ agents that are closer to $w_1$ than to $o$. Hence, we have
\begin{align*}
\SC(\bw) 
&= \sum_{i \in S} d(i,w_1) + \sum_{i \not\in S} d(i,w_1) + \sum_{i \in N} d(i,o) + \\
&\leq  \sum_{i \in S} d(i,o) + \sum_{i \not\in S} \bigg( d(i,o) + d(w_1,o) \bigg) + \sum_{i \in N} d(i,o) \\
&= \SC(\opt) +  (1-\alpha)n \cdot d(w_1,o)
\end{align*}
and
\begin{align*}
\SC(\opt) \geq 2\cdot \alpha n \cdot \frac{d(w_1,o)}{2} = \alpha n \cdot d(w_1,o).
\end{align*}
Therefore, the approximation ratio is at most $1 + \frac{1-\alpha}{\alpha} = \frac{1}{\alpha} = 1+\sqrt{2}$. 

\medskip
\noindent 
{\bf Case 2: $w_1 < o < w_2$.} \\
By the definition of the mechanism, there is a set $S_1$ of $\alpha n$ agent that are closer to $w_1$ than to $o$, i.e., $d(i,w_1) \leq d(i,o)$ for every $i \in S_1$, and thus $d(i,o) \geq d(o,w_1)/2$. Similarly, there is another set $S_2$ of $\alpha n$ agents that are closer to $w_2$ than to $o$, i.e., $d(i,w_2) \leq d(i,o)$ for every $i \in S_2$, and thus $d(i,o) \geq d(o,w_2)/2$.
By combining these facts with the triangle inequality, we have
\begin{align*}
\SC(\bw) 
&= \sum_{i \in S_1} \bigg( d(i,w_1) + d(i,w_2) \bigg) 
+ \sum_{i \in S_2} \bigg( d(i,w_1) + d(i,w_2) \bigg) 
+ \sum_{i \not\in S_1 \cup S_2} \bigg( d(i,w_1) + d(i,w_2) \bigg) \\
&\leq 
\sum_{i \in S_1} \bigg( 2\cdot d(i,o) + d(o,w_2) \bigg) 
+ \sum_{i \in S_2} \bigg( 2\cdot d(i,o) + d(o,w_1) \bigg) \\
&\quad + \sum_{i \not\in S_1 \cup S_2} \bigg( 2\cdot d(i,o) + d(o,w_1) + d(o,w_2) \bigg) \\
&= 
\SC(\opt) + (1-\alpha) n \bigg( d(o,w_1) + d(o,w_2) \bigg).
\end{align*}
We can also bound the optimal social cost as follows:
\begin{align*}
\SC(\opt) \geq 2 \cdot \alpha n \frac{d(o,w_1)}{2}  + 2 \cdot \alpha n \frac{d(o,w_2)}{2}
= \alpha n \bigg( d(o,w_1) + d(o,w_2) \bigg)
\end{align*}
Consequently, the approximation ratio is at most $1 + \frac{1-\alpha}{\alpha} = 1/\alpha = 1+\sqrt{2}$. 

\medskip
\noindent
{\bf Case 3: $o < w = w_1 = w_2$ (the case $w_1 = w_2 = w < o$ is symmetric).} \\
By the definition of the mechanism, there is a set $S$ of $(1-\alpha)n$ agents that are closer to $w$ than to $o$. Hence, by the triangle inequality, we have 
\begin{align*}
\SC(\bw) 
&= 2\cdot \sum_{i \in S} d(i,w) + 2 \cdot \sum_{i \not\in S} d(i,w) \\
&\leq 2\cdot \sum_{i \in S} d(i,o) + 2 \cdot \sum_{i \not\in S} \bigg( d(i,o) + d(o,w) \bigg) \\
&= \SC(\opt) + 2 \alpha n \cdot d(o,w).
\end{align*}
and 
\begin{align*}
\SC(\opt) \geq 2\cdot (1-\alpha)n \cdot \frac{d(o,w)}{2} = (1-\alpha)n \cdot d(o,w)
\end{align*}
Therefore, the approximation ratio is at most $1 + \frac{2\alpha}{1-\alpha} = \frac{1+\alpha}{1-\alpha} = 1+\sqrt{2}$.

\medskip
\noindent
{\bf Case 4: $o < w_1 < w_2$ (the case $w_1 < w_2 <o$ is symmetric).} \\
Clearly, since $o < w_1 <w_2$, $d(o,w_2) = d(o,w_1) + d(w_1,w_2)$. 
By the definition of the mechanism, there is a set $S$ of $(1-\alpha)n$ agents who are closer to $w_1$ than to $o$, i.e., $d(i,w_1) \leq d(i,o)$ for every $i \in S$, and thus $d(i,o) \geq d(o,w_1)/2$. Also, there is a set $T \subset S$ of $\alpha n$ agents who are closer to $w_2$ than to $w_1$, i.e., $d(i,w_2) \leq d(i,w_1) \leq d(i,o)$ for every $i \in T$, and thus $d(i,o) \geq d(o,w_1) + d(w_1,w_2)/2$. 
By combining these two facts with the triangle inequality, we have
\begin{align*}
\SC(\bw) 
&= \sum_{i \not\in S} \bigg( d(i,w_1) + d(i,w_2) \bigg) +  \sum_{i \in S \setminus T} \bigg( d(i,w_1) + d(i,w_2) \bigg)
+  \sum_{i \in T} \bigg( d(i,w_1) + d(i,w_2) \bigg) \\
&\leq \sum_{i \not\in S} \bigg( 2\cdot d(i,o)  + d(o,w_1) + d(o,w_2) \bigg) 
+  \sum_{i \in S \setminus T} \bigg( 2\cdot d(i,o) + d(o, w_2) \bigg)
+  2\cdot \sum_{i \in T} d(i,o) \\
&= \SC(\opt) + \alpha \bigg( d(o,w_1) + d(o,w_2) \bigg) + (1-2\alpha) d(o,w_2) \\
&= \SC(\opt) + \alpha d(o,w_1) + (1-\alpha) d(o,w_2) \\
&= \SC(\opt) + d(o,w_1) + (1-\alpha) d(w_1,w_2). 
\end{align*}
For the optimal social cost, we have
\begin{align*}
\SC(\opt) 
&\geq 2 |S\setminus T| \frac{d(o,w_1)}{2} + 2 |T| \bigg( d(o,w_1) + \frac{d(w_1,w_2)}{2} \bigg) 
= d(o,w_1) + \alpha d(w_1,w_2). 
\end{align*}
Hence, the approximation ratio is at most $1 + \frac{1-\alpha}{\alpha} = 1/\alpha = 1 + \sqrt{2}$. 
\end{proof}

For general instances we show a tight bound of $3$. The lower bound follows by the fact that when all agents have singleton preferences, then the problem reduces to two independent single-facility location problems, and the best possible approximation ratio for each of them is $3$~\citep{feldman2016voting}; alternatively, one can verify that the proof of Theorem~\ref{thm:sc-general-lower} holds even when the facilities can be placed at the same location. For the upper bound, we consider the {\sc Two-Medians} mechanism, which independently places each facility $F_j$ at the location closest to the median agent $m_j \in N_j$.

\begin{theorem}
For general instances, the approximation ratio of {\sc Two-Medians} is at most $3.$
\end{theorem}

\begin{proof}
Using the fact that the median agent $m_j$ minimizes the total distance of all the agents in $N_j$, the fact that $w_j = t(m_j)$, and the triangle inequality, we have
\begin{align*}
\SC(\bw) 
&= \sum_{j \in [2]} \sum_{i \in N_j} d(i,w_j) \\
&\leq \sum_{j \in [2]} \sum_{i \in N_j} \bigg( d(i,m_j) + d(m_j,w_j) \bigg) \\
&\leq \sum_{j \in [2]} \sum_{i \in N_j} \bigg( d(i,m_j) + d(m_j,o_j) \bigg) \\
&\leq \sum_{j \in [2]} \sum_{i \in N_j} \bigg( 2\cdot d(i,m_j) + d(i,o_j) \bigg) \\
&\leq 3 \cdot \sum_{j \in [2]} \sum_{i \in N_j} d(i,o_j),
\end{align*}
and thus the approximation ratio is at most $3$. 
\end{proof}

\subsection{Max cost}\label{sec:same:max}
We now consider the max cost and start by showing a tight bound of $2$ for doubleton instances. The lower bound follows by a sequence of instances similar to those in the proof of Theorem~\ref{max:both-lower-2} but just with two candidate locations. 

\begin{theorem} \label{max:both-lower-2-same-location}
For doubleton instances, the approximation ratio of any deterministic strategyproof mechanism is at least $2-\delta$, for any $\delta > 0$. 
\end{theorem}

\begin{proof}
Consider an arbitrary deterministic mechanisms and the following instance $I_1$: There are two candidate locations at $-1$ and $1$ and two agents (that approve both facilities) positioned at $-\varepsilon$ and $\varepsilon$, respectively, for some infinitesimal $\varepsilon > 0$. 

First, suppose that the mechanism places both facilities at one of the two locations, say $-1$. Then, consider the instance $I_2$ in which the agent at $\varepsilon$ in $I_1$ moves to $1$ in $I_2$, while the other agent remains at $-\varepsilon$. The mechanism must still place both facilities at $-1$ in $I_2$ since otherwise the agent that moved would decrease her cost. However, $\MC(-1,-1) \approx 6$ and $\MC(1,1) \approx 2$, leading to an approximation ratio of at least $3$. 

Second, support that the mechanism places one facility at $-1$ and the other at $1$. Then, consider the instance $I_3$ in which the agent at $\varepsilon$ in $I_1$ moves to $2$ in $I_3$, while the other agent remains at $-\varepsilon$. The mechanism must either still output the solution $(-1,1)$ or the solution $(-1,-1)$, but it cannot output $(1,1)$ as then the agent that moved would decrease her cost. However, $\MC(-1,1) \approx 4$, $\MC(-1,-1) = 6$, and $\MC(1,1)=2$, leading to an approximation ratio of at least $2$.  
\end{proof}

For the upper bound, we consider the mechanism that places $F_1$ at the candidate location closest to the leftmost agent $\ell$ and $F_2$ at the candidate location closest to the rightmost agent $r$. We refer to this mechanism as {\sc Leftmost-Rightmost}; see Mechanism~\ref{mech:max-leftmost-rightmost}. 

\SetCommentSty{mycommfont}
\begin{algorithm}[h]
\SetNoFillComment
\caption{\sc Leftmost-Rightmost}
\label{mech:max-leftmost-rightmost}
{\bf Input:} Reported positions of agents with doubleton preferences\;
{\bf Output:} Facility locations $\bw = (w_1,w_2)$ \;
$\ell \gets$ leftmost agent in $N_1 \cap N_2$\;
$r \gets$ rightmost agent in $N_1 \cap N_2$\;
$w_1 \gets t(\ell)$\;
$w_2 \gets s(r)$\;
\end{algorithm}

\begin{theorem} \label{thm:maxL-sp}
For doubleton instances, the approximation ratio of {\sc Leftmost-Rightmost} is at most $2$.
\end{theorem}

\begin{proof}
Let $i \in \{\ell,r\}$ be the agent that determines the max cost of the mechanism, and $j \in \{\ell, r\} \setminus\{i\}$. 
Let $\opt = (o_1,o_2)$ be an optimal solution. Since $w_1 = t(\ell)$ and $w_2 = t(r)$, by the triangle inequality and the definition of $t(\cdot)$, we have
\begin{align*}
\MC(\bw) &= d(i,t(i)) + d(i,t(j)) \\
&\leq d(i,t(i)) + d(i,o_2) + d(j,o_2) + d(j,t(j)) \\
&\leq d(i,o_1) + d(i,o_2) + d(j,o_2) + d(j,o_1) 
\leq 2\cdot \MC(\opt). 
\end{align*}
Therefore, the approximation ratio is at most $2$ in any case. 
\end{proof}

For general instances, it is not hard to obtain a tight upper bound of $3$. The lower bound follows again by the fact that with singleton preferences the problem is equivalent to two independent single-facility location problems, while the upper bounds follows by the variant of the {\sc Leftmost} mechanism that places $F_j$ at the leftmost agent $\ell_j \in N_j$. 

\begin{theorem}
For general instances, the approximation ratio of {\sc Leftmost} is at most $3$. 
\end{theorem}

\begin{proof}
Let $i$ be the agent that determines the max cost of the mechanism. 
By the triangle inequality and the definition of $t(\cdot)$, we have
\begin{align*}
\MC(\bw) = \sum_{j \in [2]: i \in N_j} d(i,w_j) 
&\leq  \sum_{j \in [2]: i \in N_j} \bigg( d(i,o_j) + d(\ell_j, o_j) + d(\ell_j,w_j) \bigg) \\
&\leq  \sum_{j \in [2]: i \in N_j} \bigg( d(i,o_j) + 2\cdot d(\ell_j, o_j) \bigg) \\
&\leq 3 \cdot \MC(\opt).
\end{align*}
Hence the approximation ratio is at most $3$. 
\end{proof}


\section{Conclusion}
In this paper we studied a truthful two-facility location problem with candidate locations and showed bounds on the best possible approximation ratio of deterministic strategyproof mechanisms in terms of the social cost and the max cost. An obvious question that our work leaves open is to close the gaps between our lower and upper bounds (for doubleton and general instances for the social cost, and doubleton instances for the max cost). Also, it would be interesting to consider the design of randomized strategyproof mechanisms with improved approximation guarantees. There are also multiple ways to extend our model. One such way is to change the assumption about what type of information is public or private, and instead of considering the case where the positions are private and the preferences are known as we did in this paper, focus on the case where the positions are public and the preferences are private (which is a generalization of the models studied by \citet{serafino2016} and \citet{kanellopoulos2021discrete}). Other ways of extending our model include settings with more than just two facilities to place, and more general metric spaces than just the line. 


\bibliographystyle{plainnat}
\bibliography{references}

\end{document}